\documentclass[12pt]{article}

\usepackage[left=1in,right=1in,top=1in,bottom=1in]{geometry}

\usepackage{pstricks, fullpage, subfigure, graphicx, caption, amsthm, amsmath, amsfonts, mathtools, latexsym, amssymb, verbatim, url, setspace, multirow, fancyhdr, wasysym, pdfsync, epstopdf, rotating, xfrac, footnote, titling, enumerate}
\usepackage{dsfont}
\usepackage{xcolor}

\usepackage[hidelinks]{hyperref}
\hypersetup{
	colorlinks=false,       
	linkcolor=red,          
	citecolor=green,        
	filecolor=magenta,      
	urlcolor=blue           
}
\usepackage{cleveref}

\usepackage{blkarray}
\usepackage{tikz}
\usepackage[round,comma]{natbib}
\usepackage[bottom]{footmisc}
\usepackage{setspace}
\usepackage{float}
\usepackage{bm}
\newtheorem{theorem}{Theorem}

\newtheorem{corollary}{Corollary}

\newtheorem{definition}{Definition}
\newtheorem{example}[theorem]{Example}

\newtheorem{lemma}{Lemma}

\newtheorem{proposition}{Proposition}

\newcommand{\noi}{\noindent}

\newcommand{\esp}{\mathbb{E}}

\newcommand{\ul}[1]{\underline{#1}}

\newcommand{\ol}[1]{\overline{#1}}

\def\EE{\mathbb{E}}

\def\NN{\mathbb{N}}

\newcommand{\ind}[1]{\mathds{1}_{#1} }

\title{Markets for Models}
\author{Krishna Dasaratha\thanks{Department of Economics, Boston University. email: {\tt krishnadasaratha@gmail.com}} \\ Juan Ortner\thanks{Department of Economics, Boston University. email: {\tt jortner@bu.edu}} \\ Chengyang Zhu\thanks{Department of Economics, Boston University. email: {\tt zhuc@bu.edu}}}
\date{\today}

\begin{document}

\maketitle
\begin{abstract}
\begin{singlespacing}
Motivated by the prevalence of prediction problems in the economy, we study markets in which firms sell models to a consumer to help improve their prediction. Firms decide whether to enter, choose models to train on their data, and set prices. The consumer can purchase multiple models and use a weighted average of the models bought. Market outcomes can be expressed in terms of the \emph{bias-variance decompositions} of the models that firms sell. We give conditions when symmetric firms will choose different modeling techniques, e.g., each using only a subset of available covariates. We also show firms can choose inefficiently biased models or inefficiently costly models to deter entry by competitors. \\


\end{singlespacing}
\end{abstract}

{\normalsize{}\thispagestyle{empty}
\setcounter{page}{0}}{\normalsize\par}
\pagebreak 
 
\section{Introduction}

Prediction problems are ubiquitous in the economy. To give a few examples, firms selling products often want to predict customers' willingness to pay and may use business analytics tools to do so. Banks want to assess the credit risk of  borrowers, and rely on predictive models to guide their decisions. In science and engineering, researchers want to predict the viability of compounds in domains ranging from drug discovery to materials science. Campaigns and observers want to predict elections, and may commission polls to do so.

Making quantitative predictions usually involves collecting relevant data and training statistical models based on that data. While firms, organizations, and individuals can conduct this modeling internally, many rely on external firms for these services. There is considerable interest in understanding competition among firms providing data or artificial intelligence services,\footnote{See, for example, \cite{korinek2025concentrating} and \cite{hagiu2025artificial}, for discussions of competition in these markets.} but much of the recent theoretical work focuses on monopoly settings. In this paper, we provide an economic theory analysis of settings where firms compete to sell prediction models. Our focus is on how market structure and welfare depend on the statistical properties of the models available to firms.

Before describing our analysis in some detail, we describe two high-level findings. First, even identical firms servicing identical consumers will often choose differentiated models. Second, firms can choose biased or excessively expensive models to deter competition or to achieve more favorable divisions of the total surplus. This generates inefficient market structures as well as inefficient choices of models. 

We consider an agent, who we call a consumer, facing a prediction problem. The consumer's utility is the negative of the mean squared error of their prediction (or a fixed outside option). To make predictions, the consumer purchases models from firms. These firms decide whether to enter, choose models to train on their datasets (which we assume are independent across firms), and set prices for the resulting predictions. We take an abstract approach to model choice that can accommodate simple models such as linear or ridge regression but also allows for more complex models. Importantly for our analysis, the consumer can use a single model but can also purchase multiple models and use a weighted average of these models.\footnote{In the context of demand for LLMs, \cite*{fradkin2025demand}  documents a number of popular apps making use of a mix of available models.}

A key part of the setup is what the consumer and competitors know about firms' models, and we assume that model choices are observable but the underlying datasets are not. The content of this assumption is that the consumer makes purchasing decisions based on knowledge of the modeling techniques used by firms and the structure of their datasets but not on information depending on the realized data points. As a consequence, market outcomes given model choices can be expressed in terms of the \emph{bias-variance decompositions} of the models that firms sell. When a single firm enters the market, it can extract all surplus. When multiple firms enter, each is paid the marginal value of combining their model with other models. These marginal values can be expressed in terms of the chosen models' biases, which are deterministic and often well-behaved, and variances, which are easy to relate under our independence assumption. Our modeling choices therefore yield straightforward mappings from the statistical structures of firms' models to outcomes and payoffs, and we use these mappings to explore market structure under various circumstances.

We begin by considering firms making simultaneous choices about whether to enter the market and what model to use. In this setting, efficient outcomes can always be supported as equilibrium outcomes: because firms are paid the marginal value of their models to the consumer, incentives can be aligned. But there can also be inefficient equilibria. Moreover, even under efficient equilibria, the division of surplus  between consumer and firms can depend in nontrivial ways on the underlying statistical properties of the available models. 

To illustrate firms' entry choices, we consider a special case where all firms share the same statistical model (and therefore the main strategic consideration is whether to enter). Pure-strategy equilibria exist, and entry decisions at these equilibria are efficient. Higher variance models induce more entry because noisier models allow scope for competitors to enter. The effect of model variance on consumer surplus is non-monotonic: a higher variance decreases total surplus but increases competition, and we show the consumer attains the highest surplus when only two firms enter the market. Overall, decreasing model noise leads to more concentrated markets, which can actually harm consumers.

Next, we analyze firms' model choices and ask when firms will choose distinct models at equilibrium. We consider a setting with a single consumer and firms with the same sets of possible models, so any differentiation must be driven by complementarities between different modeling techniques. We derive necessary conditions for firms to choose the same model and decompose these conditions into interpretable forces favoring and opposing differentiation. We provide examples of differentiation when firms use linear regression and choose which covariates to include. We show firms can choose distinct subsets of covariates for the regressions at equilibrium---even if each firm could costlessly include all covariates in the true model. This corresponds to different companies offering predictive models based on different types of data.

Finally, we consider a setting with an incumbent firm and a potential challenger and describe strategies the incumbent can use to deter entry. The incumbent chooses a model, and the challenger then decides whether to enter the market and choose a model. First, we show that the incumbent can bias their model to deter entry. Intuitively, if available modeling techniques share common biases, an incumbent firm can have little incentive to correct those biases. Second, we show that the incumbent can overinvest in reducing model variance. A natural interpretation is the incumbent firm acquiring a large dataset, e.g., artifical intelligence firms building `data moats' to protect market position. Both behaviors lead to both inefficient entry and inefficient model choices. The equilibrium markets are monopolies although competitive markets would generate more surplus, and the incumbent chooses a more biased or more expensive model than would be optimal given a monopoly.


The rest of this paper is organized as follows. The remainder of this section reviews the existing literature relevant to this paper.  \Cref{s:Setup} formally describes our setup. \Cref{s:simultaneous} studies the case where firms enter simultaneously, including analyzing when equilibrium features model differentiation. \Cref{s:sequential} looks at the problem with sequential entry and analyzes strategies to deter entry. \Cref{s:conclusion} concludes. All omitted proofs are deferred to \Cref{proofs}.

\paragraph{Related Literature}

Research on the economics of providing predictions, models, and forecasts spans fields including economics, finance, and computer science. Like this paper, some recent computer science work has considered agents or firms competing to provide predictions (e.g., \cite{ben2017best}, \cite{ben2019regression}, \cite*{feng2022biasvariancegames}, \cite*{jagadeesan2024improved}, and \cite*{jagadeesan2024safety}). A bit further afield, recent research considers agents competing to provide algorithms to consumers in other domains, including ranking alternatives (e.g., \cite*{immorlica2011dueling} and \cite{kleinberg2021algorithmic}) and reinforcement learning (e.g., \cite*{aridor2020competing}).

We highlight three contributions relative to other work on markets for prediction models. First, we provide a microfounded framework in which equilibrium outcomes depend straightforwardly on the bias-variance decompositions of firms' chosen models. Relatedly, \cite{feng2022biasvariancegames} take a class of models with a bias-variance tradeoff as a primitive in a contest model. They find competition pushes toward higher variance and lower bias models. Second, we ask what market structures emerge in a variety of settings with endogenous entry. Perhaps closest, \cite{jagadeesan2024safety} also consider entry in a model focused on regulatory constraints. Third, we allow consumers the possibility of purchasing models from several firms and find this can lead to quite different market structures.

Several strands of research in economics and finance consider agents competing to make predictions or forecasts, with a focus on contests and related games. \cite*{montiel2022competing} consider auctions when bidders use subjective models and evaluate prediction error using those models (and not as correctly specified Bayesians).\footnote{Several recent papers study settings in which decision makers' beliefs are derived from (potentially) misspecified models \cite*[e.g.,][]{spiegler2016bayesian, eliaz2020model, levy2022misspecified}. Relatedly, \cite*{izzo2023ideological} and \cite*{montiel2025competing} study settings in which political parties compete by providing statistical models that help explain past data.} \cite{ottaviani2006strategy}, \cite{ottaviani2007aggregation}, and a subsequent financial economics literature consider forecasters in contests and in cheap talk games. We instead focus on firms selling models to consumers and ask about the resulting market structure. A common thread is that there can be strategic reasons to provide biased models. We find the relevant strategic forces are different in markets than in contests, however, e.g., choosing biased models to deter entry.

Our work also relates to two rapidly growing strands of literature: the literature that studies markets for information \citep*{bergemann2018design, bergemann2019markets, yang2022selling, galperti2024value}, and the literature that studies the economics of large language models \citep*{duetting2024mechanismdesignlargelanguage, mahmood2024pricingcompetitiongenerativeai, kumar2025pricing, bergemann2025economicslargelanguagemodels}. Our contribution relative to these strands of literature is to propose a framework to study competition among model providers.  

Lastly, our work relates to another relevant literature that considers sender-receiver games with multiple senders \citep{milgrom1986relying, gentzkow2016competition, gentzkow2017bayesian, li2021sequential}. These papers focus on settings in which the incentives of senders and receiver are misaligned, and study how competition influences the information that senders disclose in equilibrium. Instead, we consider a setting in which the incentives of model providers and the consumer are aligned, and focus on how market structure depends on the statistical properties of the models available to firms.  

\section{Basic Setup}\label{s:Setup}

There are $N\in \NN$ firms and a consumer. The consumer's problem is to predict 
$$y= f(x) + \epsilon$$
where $f: \mathbb{R}^k \rightarrow \mathbb{R}$ is measurable and $\epsilon$ is an i.i.d., mean zero error term with variance $\sigma^2$. The distribution of the noise term $\epsilon$ and the distribution $G(\cdot)$ of the vector of covariates $x$ are common knowledge among firms and the consumer. The function $f$ is drawn from a common knowledge distribution $F(\cdot)$.

Each firm $i$ chooses a model and then trains this model on an independent dataset. A dataset $D^{(i)}$ is a finite set of $n$ independent data points. Data points are $(x_j,y_j) \in \mathbb{R}^{k+1}$, where $x_j$ are drawn independently from distribution $G(\cdot)$, and $y_j = f(x_j) + \epsilon_j$. We call the set of possible datasets $\mathcal{D}$.\footnote{In Section \ref{s:simultaneous} we briefly discuss how our results generalize when firms have correlated datasets.}

A model is a measurable function $M: \mathbb{R}^k \times \mathcal{D}  \rightarrow \mathbb{R}$. Given a dataset $D \in \mathcal{D}$, each model $M$ defines an estimator $\hat{f}_D: \mathbb{R}^k \rightarrow \mathbb{R}$ by $\hat{f}_D(x) = M(x,D)$. The consumer can buy at most one model from each firm in the market to help with her prediction.


We study two versions of our game: one in which firms choose models simultaneously (Section \ref{s:simultaneous}) and one in which firms choose models sequentially (Section \ref{s:sequential}). In both versions,  each firm $i$ chooses an action $M_i \in \mathcal M_i \cup \emptyset$, where $\mathcal{M}_i$ is the set of available models to firm $i$ and $\emptyset$ denotes not entering the market. We allow entry costs to vary with the model that firms choose: the cost of choosing model $M \in M_i$ is $c(M)$, with $c:\bigcup_i \mathcal M_i \rightarrow \mathbb R_{++}$, while the cost of choosing action $\emptyset$ is 0.

After firms choose their models, all firms publicly observe the model choices of all their competitors, and simultaneously set prices $(p_i)$. 

The consumer observes the models and prices of the firms and decides which models to buy.   
Call the set of models purchased by the consumer $P \subseteq \{1,2,\hdots,N\}$. To form a prediction, the consumer combines the models she bought, and we assume that the consumer does so by choosing non-negative weights $w^P \in \{w \in \mathbb R_+^{|P|} : \sum_{i \in P} w_i = 1\}$ for the models in $P$. 
In the interim stage, each firm $i$ privately observes its dataset $D^{(i)}$, and the consumer learns the vector $x$. Each firm $i \in P$ gives the consumer their prediction $\hat f_{i,D^{(i)}}(x) = M_i(x,D^{(i)})$. The consumer's prediction is then
$$\phi(\{\hat{f}_{i,D^{(i)}}(x)\}_{i \in P}) := \sum_{i\in P} w_i^P \hat{f}_{i, D^{(i)}}(x).$$

The consumer's payoff is equal to the negative of the mean squared error in her prediction, minus the prices of the models she bought. Hence, the consumer's expected payoff from buying models in $P$ and using prediction $\phi(\{\hat{f}_{i,D^{(i)}}(x)\}_{i \in P})$ is
$$-\esp_{f,D,x,\epsilon} \left[\left(\phi(\{\hat{f}_{i,D^{(i)}}(x)\}_{i \in P}) - y \right)^2 \right] - \sum_{i \in P} p_i.$$
We assume that the consumer's beliefs about the true model $f$, datasets $D = (D^{(j)})_{j=1}^N$, vector $x$ and noise $\epsilon$ are equal to the prior for any models that the firms choose; that is, the consumer doesn't update her beliefs after observing firms' model choices. We denote by $\ul u<0$ consumer's utility from taking the outside option, which represents the expected payoff from not buying any models and making the prediction on her own.

A firm that enters the market with model $M_i$ earns profits $p_i  \ind{i \in P}  - c(M_i)$, and a firm that does not enter the market earns zero profits. Our solution concept is subgame perfect Nash equilibrium (SPNE).




We now present a preliminary result that will be useful in our analysis and several illustrative examples. This first result is a bias-variance decomposition of the mean squared loss, which is a standard result in statistics and machine learning (see, for example, \cite*{james2023introduction}). 

\begin{lemma}\label{p:bias_variance}
Suppose model $M$ gives estimator $\hat{f}_D(x)$ for each dataset $D$. The mean squared loss can be decomposed as:
$$\mathbb{E}_{f,D,x,\epsilon}[(\hat{f}_D(x)-y)^2] = \EE_{f,x}\left[(\underbrace{\mathbb{E}_{D}[\hat{f}_D(x) \mid f,x]-f(x)}_{\text{bias}})^2 \right]+ \mathbb{E}_{f,x}[\underbrace{\mathbb{E}_D[(\mathbb{E}_{D}[\hat{f}_D(x)\mid f, x]-\hat{f}_D(x))^2]}_{\text{variance}}]+ \sigma^2. $$
\end{lemma}

We now provide several examples of prediction problems and classes of models $\mathcal{M}$, along with the bias-variance decomposition for each.

\begin{example}\label{eg:forecasting}
Suppose data points are signals $y = \theta + \epsilon$ about a Gaussian state $\theta \sim \mathcal{N}(0, \nu^2)$ with Gaussian noise $\epsilon \sim \mathcal{N}(0,\sigma^2)$ This is a special case of our framework where the functions $f = \theta$ are constant and the prior $F(\cdot)$ is Gaussian.

If each firm observes a single data point, this recovers a standard model of forecasting (e.g., \cite{ottaviani2006strategy}). A standard class of models to consider is then $$\mathcal{M} = \{ry: r \in [0,1]\}.$$ This corresponds to taking a weighted average of the prior belief about $\theta$ (which we have normalized to $0$) and the signal $y$. The model $ry$ has bias $(1-r)\theta$ and variance $r^2\sigma^2$.
\end{example}

\begin{example} \label{eg: ols}
Suppose the true models $f(x) = x^T \beta$ are linear. Then a natural model is ordinary least squares (OLS), which corresponds to $\hat{f}_D(x) = x^T\hat{\beta}$ where $\hat{\beta}$ minimizes $$\sum_{(x_j,y_j) \in D} |y_j-x_j^T\hat{\beta}|^2.$$
In the classical regime $k<n-1$, the estimator is $\hat{\beta} = (X^TX)^{-1}X^TY,$ where $X$ is the matrix of covariates and $Y$ is the vector of outcomes in the dataset $D$. This model has zero bias and variance   (conditional on $\beta$ and $x$) equal to $\sigma^2 x^T \mathbb{E}_D [(X^TX)^{-1}] x$.
\end{example}

\begin{example} \label{eg: ridge}
Suppose the true models  $f(x) =x^T \beta$ are linear. Another widely used model is ridge regression,\footnote{For an in-depth analysis of the ridge case see \cite{rizzi2025value}, who considers a setting where the Bayesian optimal model is a ridge regression and analyzes the value of data under this model.} which corresponds to $\hat{f}_D(x) =x^T \hat{\beta}$ where $\hat{\beta}$ minimizes $$\sum_{(x_j,y_j) \in D} |y_j-x_j^T\hat{\beta}|^2 + \lambda\|\hat{\beta}\|_2^2.$$
Increasing the penalization parameter $\lambda$ leads to lower variance but higher bias. The estimator is $$\hat{\beta} = (X^TX+\lambda I)^{-1}X^TY.$$
Writing $X = U\Sigma V^T$ for the singular value decomposition of $X$ (so $U$ and $V$ are orthogonal matrices and $\Sigma$ is a diagonal matrix) and $\sigma_i = \Sigma_{ii}$, a standard calculation shows that the estimator is $$\hat{\beta} = V \mathrm{diag}\left(\frac{\sigma_1^2}{\sigma_1^2+\lambda},\hdots,\frac{\sigma_k^2}{\sigma_k^2+\lambda}\right) V^T\beta +  V \mathrm{diag}\left(\frac{\sigma_1}{\sigma_1^2+\lambda},\hdots,\frac{\sigma_k}{\sigma_k^2+\lambda}\right) U^T(Y- X \beta).$$
If the distribution of covariates $G(\cdot)$ is rotationally invariant, then we can write the squared bias conditional on $x$ and $\beta$ as $$\left(1-\EE_{\sigma_j}\left[\frac{\sigma^2_j}{\sigma^2_j+\lambda}\right]\right)^2 (x^T\beta)^2,$$
where the expectation is taken over the singular values $\sigma_j$ of the random matrix $X$. Thus squared bias is increasing in $\lambda$.


\end{example}

One feature of \Cref{eg:forecasting} and \Cref{eg: ridge} with rotationally invariant covariates is that the choice of model affects the magnitude of biases but not the direction. This property will be useful as a special case throughout and as an assumption for one of our results. We say models have a \emph{common bias direction} if for each model $M \in \bigcup_i \mathcal{M}_i$,
\begin{equation}\label{eq:common_bias}
\mathbb{E}_{D}[\; \hat{f}(x)-f(x)\mid f,x \; ] = \alpha b_0(f,x)
\end{equation}
where $\alpha > 0$ can depend on $M$  but $b_0(f,x)$ does not.

\subsection{Discussion}

Before turning to our analysis we briefly discuss the setup and our assumptions. We study settings where firms sell models rather than directly selling data. We note two reasons why this often occurs. First, consumers may lack resources or domain knowledge to train models internally. Second, because data is non-rival, firms have incentives to protect their raw data to maintain market power \citep{jones2020nonrivalry}.

Relatedly, we allow firms to choose models from a restricted set rather than modeling firms as unrestricted Bayesian agents. This lets our framework accomodate a range of modeling techniques used in practice.

We assume that the consumer makes purchasing decisions based on observing firms' modeling choices but not based on any information about the realized datasets. Our observability assumptions fit well if firms can communicate their modeling techniques but want to keep their datasets proprietary. An alternative approach, which is more challenging to work with in most settings, would be to let consumers calculate prediction errors given realized datasets.

We let consumers combine several models but restrict them to choosing weighted averages of these models. If the consumer could choose  arbitrary functions of the predictions purchased from firms, then they would adjust the predictions to remove biases. We remove such concerns, which would lead to highly sophisticated behavior that we do not see as matching the spirit of our approach. We do let the consumer choose the optimal weights to average predictions across firms.

Finally, we note that a firm's choice of model in our framework can capture both decisions about the modeling technique (formally represented by $M$) and decisions about the structure of its dataset. We focus on the former in our analysis. But for some discussions and examples, we allow the size of firms' datasets to depend on the model that they choose: as part of their entry decision, firms may choose how many data points, or how many covariates, to collect.

\section{Simultaneous Entry}\label{s:simultaneous}

We now study the game we described above, in the case in which all firms choose their models simultaneously. In particular, in  the first stage of the game, each firm $i$ simultaneously chooses an action in $\overline{ \mathcal{M}}_i \equiv \mathcal{M}_i \cup \emptyset$. In the second stage, firms' model choices are publicly observed, and all firms that entered the market simultaneously set prices. The consumer then observes the firms' models and prices, and chooses which models to purchase and their weights. Lastly, all uncertainty is resolved: each firm $i$ privately observes its dataset $D^{(i)}$; and the consumer learns vector $x$, receives prediction $\hat f_{i,D^{(i)}}(x) = M_i(x,D^{(i)})$ from each model $i$ she bought, and combines them according to the chosen weights.

We begin by describing the equilibrium pricing behavior of the firms. Let $\mathbf{M} = (M_i)_{i=1}^N$ be the profile of actions chosen by the firms in the first stage. Let $E(\mathbf{M}) \equiv \{i : M_i \neq \emptyset\}$ be the set of firms that enter the market and $N_E = |E(\mathbf{M})|$ be the number of entrants. 
For any non-empty subset $E' \subset E(\mathbf M)$, let $U(E',\mathbf M)$ denote the negative of the expected square loss from purchasing models $(M_i)_{i \in E'} = (\hat f_{i,D^{(i)}})_{i \in E'}$ and weighting them optimally: 
$$U(E',\mathbf M) = - \mathbb{E}\left[\left( \sum_{j\in E'}w_j^{E'}\hat{f}_{j,D^{(j)}} (x) -y\right)^2\right].$$
We also let $U(\emptyset ,\mathbf M) = \ul u$ be the consumer's payoff from not buying any models. 


\begin{proposition}\label{p:prices1}
Suppose firms choose models $\mathbf{M} = (M_i)_i$. 
Then, the subgame that starts at the pricing stage has a SPNE in which the consumer purchases all models in $E(\mathbf{M})$, with prices satisfying
\begin{align} 
\forall i \in E(\mathbf M), \quad p_i &= \min_{E' \subset E(\mathbf M) \backslash \{i\}} \left[  U( E(\mathbf{M}), \mathbf M) - U(E', \mathbf{M}) - \sum_{\substack{j \in E(\mathbf{M}) \backslash E' \\ j \neq i}} p_j \right]. \label{e:prices}
\end{align}
Moreover, in every SPNE of the pricing subgame under which the consumer buys all models  in $E(\mathbf{M})$, prices must satisfy \eqref{e:prices}. 
\end{proposition}

\Cref{p:prices1} characterizes the prices that firms charge when the consumer buys all models in $E(\mathbf M)$. We note that, since firms face strictly positive entry costs, on the equilibrium path of every pure strategy SPNE the consumer buys all models in $E(\mathbf M)$.

To understand equation \eqref{e:prices}, consider first the case in which two or more firms enter the market. The utility that the consumer gets from purchasing all models is 
$$ U( E(\mathbf{M}), \mathbf M) - \sum_{j \in E(\mathbf M)} p_j,$$
while the maximum utility that the consumer can get from not purchasing firm $i$'s model is
$$\max_{E' \subset E(\mathbf M) \backslash \{i\}} \left[U( E', \mathbf M) - \sum_{j \in E'} p_j \right].$$
The consumer is willing to buy firm $i$'s model if
$$\min_{E' \subset E(\mathbf M) \backslash \{i\}} \left[ U( E(\mathbf{M}), \mathbf M) - U(E', \mathbf{M}) - \sum_{\substack{j \in E(\mathbf{M}) \backslash E' \\  j \neq i}} p_j - p_i\right] \geq 0 .$$
The price $p_i$ in \eqref{e:prices} leaves the consumer indifferent between buying all models and not purchasing $i$'s model. 

When there is one entrant, that firm extracts all the surplus from the consumer by charging a price of $U( E(\mathbf{M}), \mathbf M) - U( \emptyset, \mathbf M) =U(E(\mathbf{M}), \mathbf M) - \ul u $.  

In general, the system of equations \eqref{e:prices} defines a fixed-point condition that can have many solutions. We next show that a natural condition simplifies this system considerably.

\begin{definition}\label{a:dmr}
We say that models satisfy \emph{decreasing marginal returns} if, for all $\mathbf M \in \prod_i \overline{\mathcal M}_i$,
$E \subset E(\mathbf M)$, $E' \subset E$ and $j \in E'$
$$U(E',\mathbf M) - U(E' \backslash \{j\},\mathbf M) \geq U(E,\mathbf M) - U(E\backslash \{j\},\mathbf M). $$
\end{definition}
In words, models have decreasing marginal returns if the marginal value for the consumer of buying an additional model is decreasing in the number of models the consumer is already buying. Note that the condition is always satisfied when there are two firms in the market (i.e., $N=2$) and the consumer's outside option is sufficiently negative. In \Cref{s:symmetric}, we show that \Cref{a:dmr} also holds when all firms have access to the same unique model and the consumer's outside option is sufficiently negative. Intuitively, models satisfy decreasing marginal returns unless (i) the outside option is high or (ii) certain sets of models are very complementary because their biases at least partially cancel each other out.

When models satisfy decreasing marginal returns, the pricing subgame has a simple equilibrium. Each firm charges a price equal to the marginal value of their model to the consumer, assuming the consumer purchases all other available models:

\begin{corollary}\label{cor:prices}
Suppose models satisfy decreasing marginal returns. Then, the subgame that starts at the pricing stage has a SPNE in which the consumer purchases all models in $E(\mathbf{M})$, with prices satisfying
\begin{align} 
\forall i \in E(\mathbf M), \quad p_i &= U( E(\mathbf{M}), \mathbf M) - U(E(\mathbf{M}) \backslash \{i\}, \mathbf{M}). \label{e:prices2}
\end{align}
\end{corollary}

Going back to the firms' entry decisions, we now show that any efficient action profile is a Nash equilibrium of the simultaneous-move game. For any models $\mathbf M$, let $TS(\mathbf M) = U(E(\mathbf M),\mathbf M) - \sum_{i \in E(\mathbf M)} c(M_i) - \ul u$ denote the total surplus (adjusted for consumer's outside option $\ul u$ so that total equilibrium surplus is nonnegative) when firms choose models $\mathbf M$.   
 
\begin{proposition}\label{p:efficiency} 
Suppose models satisfy decreasing marginal returns, and let $\mathbf{M}^*$ be an action profile that maximizes total surplus $TS(\mathbf M)$. Then, there exists a pure-strategy SPNE in which firms choose actions $\mathbf{M}^*$.   
\end{proposition}

When models have decreasing marginal returns, each firm charges a price equal to its marginal contribution to total surplus (see equation \eqref{e:prices2}). Hence, firms have incentives to choose models that maximize total surplus.    

As a Corollary of \Cref{p:efficiency}, we get the following existence result. 

\begin{corollary}\label{cor:existence}
Suppose models satisfy decreasing marginal returns. Then, the set of pure-strategy SPNE is non-empty if $\displaystyle \arg\max_\mathbf{M} TS(\mathbf{M})$ is non-empty.
\end{corollary}

For the remainder of \Cref{s:simultaneous}, we focus on settings in which models satisfy decreasing marginal returns and therefore  \Cref{cor:existence} holds. 

\paragraph{Correlated datasets.} While our model assumes that firms' datasets are independent, we note that our results in this section do not depend on this assumption. With correlated datasets, equilibrium prices continue to satisfy equation \eqref{e:prices} (and, if models satisfy
decreasing marginal returns, equation \eqref{e:prices2}). However, as the analysis below shows, the assumption of independent datasets greatly simplifies the analysis by allowing us to compute the consumer's expected utility $U(E,\mathbf M)$ from buying and combining models.

\subsection{Entry with a Common Model}\label{s:symmetric} 

To illustrate the workings of our model, we consider a simple symmetric setting in which all firms have access to one statistical model: for all $i$, $\mathcal M_i = \{M\}$, though each firm $i$ observes an independent dataset $D^{(i)}$. Hence, the only decision of the firms is whether or not to enter. We let $c>0$ denote the cost of model $M$.

Let $B = \EE_{f,x}\left[\left(\mathbb{E}_{D}[\hat{f}_D(x)-f(x)\mid f ,x]\right)^2 \right]$ and
$V = \esp_{f,D,x}\left[\left(\esp_{D} \left [\hat f_D(x)\mid f ,x\right] - \hat f_D(x)\right)^2\right]$
denote, respectively, the expected squared bias and variance of the model $M(x,D) = \hat f_D(x)$  available to all firms. We assume that the number $N$ of potential entrants is large, with $\frac{V}{N(N-1)} <c$. We also assume that the consumer's outside option is low enough so that at least one firm enters the market, which holds when additionally $ \ul u<-B - V - \sigma^2 -c $, and low enough to ensure decreasing marginal returns, which we show in the appendix holds when $\ul u< - B - \frac{3}{2}V - \sigma^2$. Decreasing marginal returns and \Cref{cor:prices} together imply that, in this setting with symmetric firms, the pricing subgame has a SPNE in which firms set prices according to \eqref{e:prices2} and the consumer purchases all models.

We define producer surplus as the sum of firms' profit, and consumer surplus as $U(P,\mathbf{M})-\sum_{i\in P}p_i - \ul u$, that is, the negative of the prediction error using the models bought, minus the price of the models, adjusted for the outside option $\ul u$ so that consumer surplus is nonnegative.


We can use this model to study how improvements in predictions affect entry and consumer surplus. We describe outcomes under pure-strategy equilibrium where firms enter if they are indifferent.\footnote{Generically, all pure-strategy equilibria are payoff equivalent. Equilibrium payoffs are not unique at the measure-zero subset of parameters where a firm is indifferent to entering, and we select an equilibrium where indifferent firms enter to simplify the proposition statement.} The effects of changes in bias are fairly straightforward, and we focus on the more interesting case of changes in variance.
\begin{proposition}\label{p:symmetric_deterministic}
Suppose all firms are symmetric and have access to the same model $M$ with expected variance $V$ and expected squared bias $B$. Then:
\begin{enumerate}
    \item[(i)] The number $N_E$ of entrants increases with $V$.
    \item[(ii)] 
    Consumer surplus attains a maximum at $V = 2c$, with $N_E=2$.
    \item[(iii)] Total surplus is decreasing in $V$.
\end{enumerate}
\end{proposition}

Buying additional models reduces the variance of the consumer's prediction because these models are trained on independent data. So when the model variance decreases, firms have a weaker incentive to enter because an additional model decreases the variance of the consumer's prediction by less. Hence, the market becomes more concentrated as models improve. 
Additional entry increases the consumer's surplus by providing more competing models. But there is also a competing effect: conditional on the number of firms in the market, a higher variance decreases consumer surplus. It turns out that consumer surplus is maximized at the smallest level of variance such that $N_E=2$. Total surplus, however, is decreasing in model variance. One way to see this is that entry decisions are efficient in the sense that they maximize total surplus, so adding noise cannot help.

\begin{figure}[!h]
    \centering
    \includegraphics[scale = 0.77]{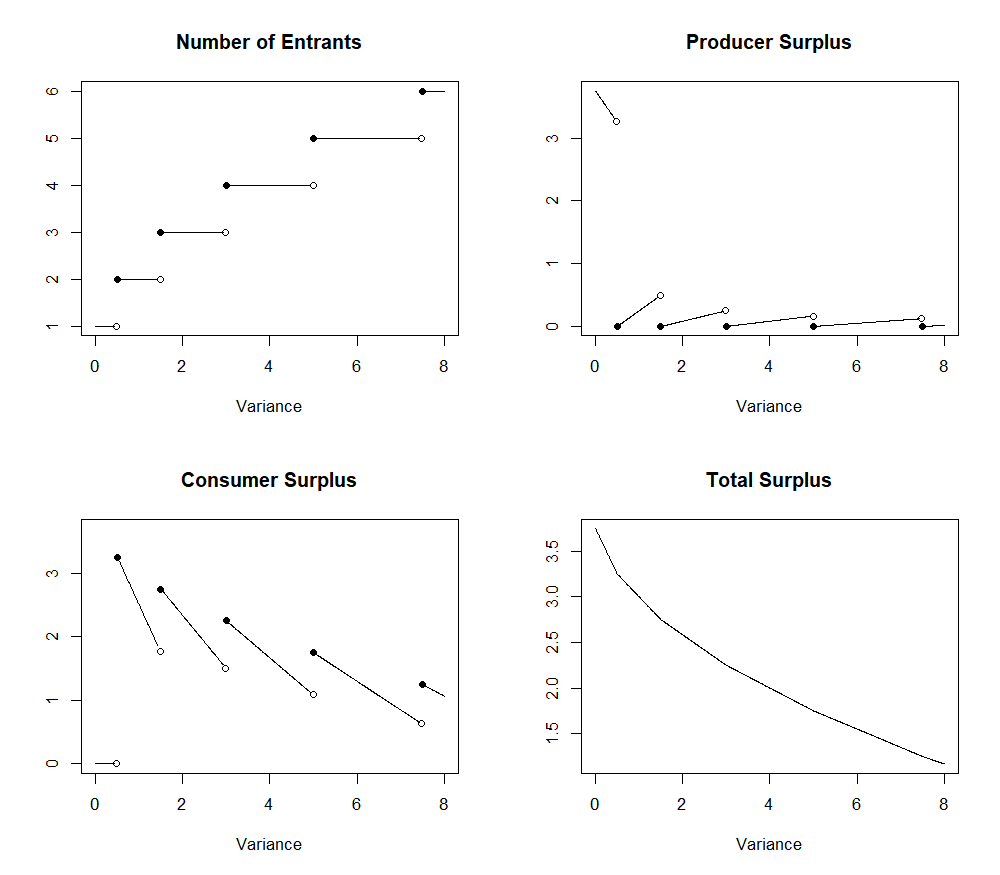}
    \caption{Equilibrium number of firms and surpluses as a function of model variance $V$.}
    \label{fig:unbiased est}
\end{figure}

To illustrate \Cref{p:symmetric_deterministic} numerically, consider the OLS model from \Cref{eg: ols}, which has zero bias. 
\Cref{fig:unbiased est} plots the number of entrants, producer surplus,  consumer surplus, and total surplus, as a function of the variance $V$ of the model available to firms. Parameter values are set to $c=0.25, \underline{u}=-5, \sigma^2=1$. The change in variance could arise, for example, from a change in the number of data points available to firms.

\subsection{Model Differentiation}

Suppose multiple firms enter and choose models from the same set $\mathcal{M}$. We now ask whether these firms will all choose the same model or will select different models. Because firms are symmetric and there is one consumer (or equivalently a population of homogeneous consumers), any differentiation must be driven by complementarities between different modeling choices.

We begin by presenting a general characterization of when firms will choose the same interior model. The characterization relies on a rich set of available models, but illustrates the main forces for and against differentiation. We then discuss several examples of model differentiation. We focus for simplicity on the case of two firms; extending the analysis to allow more firms is straightforward. The following proposition also assumes that all models have cost $c(M) = c$, but can similarly be extended to allow heterogeneous model costs.

Suppose the two firms choose from the same set of models $\mathcal{M}$ and that this set of models contains a one-parameter family  $\{M(t)\}_{t \in (\underline{t} ,\overline{t})} \subset \mathcal{M}$ for some $\underline{t}<\overline{t}$. Let $V(t)$ be the expected variance of $M(t)$. Writing $\hat{f}_{D,t}(x)$ for the estimator defined by $M(t)$, let $ b_t(f,x) = \mathbb{E}_{D}[\hat{f}_t(x) \mid f,x]-f(x)$ be the conditional bias given $f(x)$ and $x$, and let $B(t) = \mathbb{E}_{f,x}[b_t(f,x)^2]$ be the expected squared bias.

A key quantity in the following result is the angle $$\cos^{-1}\left( \frac{\int b_t(f,x)b_{t'}(f,x) dF dG}{\sqrt{B(t)B(t')}}\right)$$
  between the conditional biases of $M(t)$ and $M(t')$. Fixing $t_0$, we write $\theta(t)$ for the angle between $M(t)$ and $M(t_0)$. We note that $\theta(t)$ is identically zero when models have a common bias direction, but can be non-zero more generally.

We assume that the expected variance $V(t)$ and the conditional bias $ b_t(f,x)$ are twice-differentiable functions of $t$. Finally, we continue to assume that the outside option is sufficiently negative for decreasing marginal returns to hold.

\begin{proposition}\label{prop:diff}
Suppose both firms choose model $M(t_0)$ in an equilibrium, where $\underline{t} < t_0 < \overline{t}$. Then
\begin{equation}\label{e:diff_FOC}
    V'(t_0) + 2 B'(t_0) = 0
\end{equation}
and 
\begin{equation}\label{e:diff_SOC}
 -\frac14\,V''(t_0)
    -\frac12\,B''(t_0)
    +B'(t_0)^{2}\left(\frac{1}{8\,B(t_0)} + \frac{1}{4 V(t_0)}\right)
    +\frac{B(t_0)}{2}\theta'(t_0)^{2}
 \leq 0.
 \end{equation}
\end{proposition}

We obtain the result by simplifying the first-order condition and second-order condition for both firms to choose model $M(t_0)$. One can show firms choose different models by ruling out boundary models and showing the left-hand side of \eqref{e:diff_SOC} is positive whenever the first-order condition holds.

The first-order condition \eqref{e:diff_FOC} is straightforward: if both firms choose model $M(t_0)$, the consumer's prediction error is $$\frac12 V(t_0) +  B(t_0).$$ The expression places more weight on bias than variance because when two firms choose the same model, the biases are identical while the variances are independent.

The left-hand side of the inequality \eqref{e:diff_SOC} can be split into three interpretable terms, which we discuss in turn:
\begin{enumerate}[(1)]
\item $\mathbf{-\frac14\,V''(t_0)
    -\frac12\,B''(t_0)
   }$: These second derivatives measure the curvature of the bias-variance tradeoff. The sign of this term, which can be positive or negative, determines whether the firms prefer both choosing $M(t_0)$ to both choosing $M(t)$ for $t \neq t_0$ nearby. If $\frac14 V(t) + \frac12 B(t)$ is concave, this rules out both firms choosing an interior model. The interesting case is therefore when this term is negative.
   
\item $\mathbf{B'(t_0)^{2}\left(\frac{1}{8\,B(t_0)}+  \frac{1}{4 V(t_0)}\right) }$: The second term, which is always non-negative, measures the benefits to differentiating between bias and variance. Because variances are independent of biases, it can be beneficial for one firm to choose a higher variance model and the other to choose a higher bias model. These benefits are larger when the bias-variance tradeoff is steeper ($B'(t)$ larger).

\item $\mathbf{\frac{B(t_0)}{2}\theta'(t_0)^{2}}$: The third term, which is also always non-negative, measures the benefits from choosing models with biases in different directions. This term is zero when models have a common bias direction, but will be positive if perturbing $t$ changes the angle between the models $M(t)$ and $M(t_0)$. When firms choose the same model, biases will be equal. But if the firms choose differentiated models with non-parallel biases, these biases will contribute less to prediction error.
\end{enumerate}

A key step in the proof uses the envelope formula to simplify the consumer's weights when we calculate firm payoffs. We can do this because the firm and consumer are aligned in their preferences over weights: both want to minimize the overall prediction error.

\medskip

We now show through an example how differentiation can arise in equilibrium. We consider the following setting:
$$y = x^T \beta + \epsilon,$$
where $x \sim N(0,I_k)$, $\epsilon \sim N(0,\sigma^2)$, and $\beta \in \mathbb R^k$. For tractability, we assume that the entries of $\beta$ are drawn i.i.d. from distribution $F_\beta$. 

Each firm observes $n$ independent data points, and the models available to firms are OLS models. In particular, firms choose the subset of covariates to include in their model: models correspond to subsets $2^{\{1,...,k\}}$. Abusing notation, we identify models $M \in \mathcal{M}$ with the corresponding subsets of covariates.

We assume that there are two firms. For $i=1,2$, let $M_i$ denote firm $i$'s model. Firm $i$'s dataset is $D^{(i)} = (Y_i,X_i)$, where $Y_i \in \mathbb R^n$ and $X_i \in \mathbb R^{n \times k}$. We assume $n>k$, so that there are more datapoints than covariates. Let $X_{i,M_i} \in \mathbb R^{n \times |M_i|}$ denote the submatrix of $X_i$ including only the coviarates in $i$'s model $M_i$. Firm $i$'s prediction is $x^T \hat \beta^{(i)}$, where the estimator $\hat \beta^{(i)} \in \mathbb R^k$ is constructed as follows: for covariates $l \in M_i$, $\hat \beta^{(i)}_{l}$ are obtained from the OLS estimator $\hat \beta^{(i)}_{M_i} = (X_{i,M_i}^T X_{i,M_i})^{-1} X_{i,M_i}^T Y_i$. For covariates $l \notin M_i$, we set $\hat \beta^{(i)}_{l} = 0$.

Finally, we consider entry costs that depend on model size: the cost of model $M$ is $c(|M|)$, where $|M|$ is the number of covariates in $M$. We now show that, in this setting, firms may choose models that include different covariates. 

\paragraph{Differentiation with costly covariates.} We start considering a setting in which the number of covariates $k$ is large and additional covariates are costly. Assume that firms' entry cost $c(|M|)$ is strictly increasing and convex, with $c(1)$ small. We now argue that firms never choose the same model when $k$ is large: there is always differentiation in equilibrium. The details are in Appendix \ref{a:diff_lr}, but the result follows from two observations. First, because $c(\cdot)$ is strictly increasing and convex, firms optimally choose models that exclude some covariates when $k$ is large. Second, if both firms choose the same model $M $ that excludes some covariates, a firm strictly profits from swapping a covariate in model $M$ for a covariate that is not included in the model. This deviation makes the biases of firms' models different, reducing consumer's prediction error. An intuition for this is that firms' models are more complementary when they include different covariates; this increased complementarity allows firms to charge higher prices. 

\paragraph{Differentiation with costless covariates.} The argument in the previous paragraph relies on the assumption that entry costs are strictly increasing in model size. This is not needed: differentiation can also arise when all models are equally costly, so additional covariates are free to include. To illustrate this, suppose we increase the number of covariates $k$ and the number of data points $n$ at the same rate. In Appendix \ref{a:diff_lr} we show that for large $k$ and $n$, when the signal-to-noise $\frac{\esp_F[\beta^2]}{\sigma^2}$ is in some intermediate range, in equilibrium firms choose models with different covariates.

These examples illustrate that market equilibria can feature firms selling models based on different parts of a data set. Indeed, this can happen even if all relevant characteristics are available to all firms at no additional cost. A number of real-world markets feature competing firms using different types of data. For example, the credit rating industry includes traditional credit rating models as well as alternative credit rating models emphasizing a range of characteristics not included in traditional models. Our analysis identifies one set of forces that could generate this market structure.

\medskip

In the examples above, differentiation increases the complementarity of firms' models and allows them to charge higher prices. We end this section by noting that this is a general result: differentiation always benefits firms at the expense of the consumer.

\begin{proposition}\label{p:differentiation_surplus}
    Suppose $\mathcal{M}_1=\mathcal{M}_2$ and consider an equilibrium with model choices $M_1 \neq M_2$. There exists $i \in \{1,2\}$ such that the consumer surplus at this equilibrium is less than or equal to consumer surplus in the equilibrium of the pricing subgame after both firms choose $M_i$.
\end{proposition}

The basic force is simple. The consumer's payoff at the equilibrium is
$$ U(\{1,2\},(M_1,M_2)) - \sum_i p_i =  U(\{1,2\},(M_1,M_2)) - \sum_i ( U(\{1,2\},(M_1,M_2)) - U(\{i\},(M_1,M_2))),$$
and the right-hand side is actually decreasing in the utility $U(\{1,2\},(M_1,M_2))$ from the equilibrium prediction. Because the consumer must pay \emph{each} firm the marginal value of their model, improving the prediction by differentiating the two models (without improving their individual prediction errors) harms the consumer.

\section{Entry Deterrence}\label{s:sequential}

We have so far studied models in which all firms choose whether to enter simultaneously. We now consider a sequential entry game in which an incumbent firm chooses a model and then a competitor decides whether to enter. This can generate several new forces, and we focus on two: the incumbent firm can choose a model with more bias to deter entry and can choose an inefficiently expensive model to deter entry.

There is an incumbent firm (firm 1), a competitor firm (firm 2), and a consumer. The consumer, firms, and informational environment continue to follow the basic setup from \Cref{s:Setup}. The incumbent firm first chooses a model $M_1 \in \mathcal{M}_1$, where $\mathcal{M}_1$ is compact. The competitor then observes this model choice and decides whether to enter. If the competitor enters, they choose a model $M_2 \in \mathcal{M}_2$,  where $\mathcal{M}_2$ is  also compact. We decompose model costs $c(M)$ into a fixed cost $c_f \in \mathbb{R}$ and a model-dependent cost $c_m: \bigcup_{i \in \{1,2\}} \mathcal{M}_i \rightarrow \mathbb{R}_{++}$ so that we can easily vary entry costs. After choosing models, the firm(s) in the market simultaneously choose prices $p_i$ for their model(s). The consumer can purchase one or both of the models or choose the outside option.

So the timing is:
\begin{enumerate}
\item Firm 1 chooses a model $M_1 \in \mathcal{M}_1$.
\item Firm 2 chooses a model $M_2 \in \mathcal{M}_2$ or chooses not to enter.
\item The firm(s) simultaneously set prices $p_i$.
\item The consumer decides which model(s) to purchase.
\end{enumerate}

We assume firm $1$'s model is sold at a positive price in any equilibrium, ruling out some cases where firm $1$'s models have much higher bias than firm $2$'s models. We also assume the outside option is sufficiently low for decreasing marginal returns to hold.

\subsection{Excessive Bias}

We will show that the incumbent firm can choose a biased model to deter entry. The basic force relies on the incumbent and entrant having models with ``similar'' biases. For simplicity, we will assume a common bias direction, and it should be clear that we could obtain the same result in cases with less structure. Recall this means that $$\mathbb{E}_{D}[\; \hat{f}(x)-f(x)\mid f,x \; ] = \alpha b_0(f,x),$$
where the right-hand side only depends on the chosen model through $\alpha$. We call $\alpha$ the bias constant corresponding to model $M$ and write $B_0 = \mathbb{E}_{f,x}[\|b_0(f,x)\|^2_2]$, so that $\alpha^2B_0$ is the expected squared bias.

\Cref{cor:prices} applies: there is a unique equilibrium of the pricing subgame under which both models are purchased, with prices given by
$$p_1 = -\mathbb{E}\left[\left(w_1^{\{1,2\}}\hat{f}_1(x)+w_2^{\{1,2\}}\hat{f}_2(x)-y\right)^2\right]+\mathbb{E}\left[\left(\hat{f}_2(x)-y\right)^2\right],$$
$$p_2 = -\mathbb{E}\left[\left(w_1^{\{1,2\}}\hat{f}_1(x)+w_2^{\{1,2\}}\hat{f}_2(x)-y\right)^2\right]+\mathbb{E}\left[\left(\hat{f}_1(x)-y\right)^2\right].$$
When both of these prices are positive, this is the unique equilibrium of the pricing subgame. When the second firm's price is zero there can be other equilibria where its model is not purchased, but these equilibria are payoff equivalent.

A consequence is that when models have a common bias direction, equilibria only depend on the bias constants $\alpha_i$ and variances $V_i= \mathbb{E}_{f,x,D}[(\mathbb{E}_{D}[\hat{f}_i(x)\mid f,x]-\hat{f}_i(x))^2]$ of models $M_i \in \mathcal{M}_i$. Identifying models with $(\alpha_i,V_i) \in \mathbb{R}^2$, we can define the \emph{Pareto frontier} of $\mathcal{M}_i$ to be the set of models that are not dominated in this space. That is, $M_i \in \mathcal{M}_i$ with bias constant and variance  $(\alpha_i,V_i)$ is contained in the Pareto frontier if there does not exist $M_i' \in \mathcal{M}_i$ with bias constant and variance $(\alpha_i',V_i')$ such that $\alpha_i' \leq \alpha_i$ and $V_i' \leq V_i$ with at least one inequality strict.

We next show that when there is a binding bias-variance tradeoff, there are parameter ranges where firm $1$ will choose a more biased model to deter entry.
\begin{proposition}\label{p:sequential_bias}
    Suppose models have a common bias direction.     Let $M_1^* \in \mathcal{M}_1$ be a model minimizing mean-squared error and $\alpha^*$ be its bias constant, and suppose the boundary of the Pareto frontier is a smooth curve with $M_1^*$ in the interior of this curve. Then there is an open subset of pairs $(c_f,\underline{u})$ such that at equilibrium, firm 1 chooses a model $M_1 \in \mathcal{M}_1$ with $\alpha_1>\alpha^*$ and firm 2 does not enter.
\end{proposition}

The proposition highlights two inefficiencies. To discuss these, it is easiest to consider the case when there is only a fixed model cost $c(M) =c_f$. First, there is inefficient model choice conditional on entry. The incumbent firm biases their model relative to the optimal monopoly model $M^*$ to maintain their market position. Second, there is inefficient entry because firm 1 deters firm 2 from entering.

The basic idea is that combining two models decreases variance (as averaging noise terms decreases variance) but does not substantially decrease bias under the common bias assumption. The competitor enters when the marginal value they provide to the consumer is high. This marginal value is high when the incumbent chooses a high variance model but low when the incumbent chooses a high bias model. Therefore, the incumbent chooses a high bias model to ensure monopoly profits. The proposition shows this occurs whenever (i) the competitor's cost $c_f$ is in an intermediate range where the incumbent can influence the entry choice, and (ii) the outside option is low enough for deterring entry to be desirable.

As a simple example, suppose the available models for each firm are ridge regressions with penalization parameters $\lambda > 0$ in some compact set. We saw in \Cref{eg: ridge} that this set of models satisfies the common bias assumption when the covariates are rotationally invariant. The proposition then says that for an open subset of parameters $c$ and $\underline{u}$, the incumbent will choose a higher $\lambda$ than would minimize mean squared error and the competitor will not enter. Because all available models bias the coefficients $\hat{\beta}$ toward zero, the incumbent can deter entry by choosing a lower variance model with smaller coefficients.

\subsection{Overinvestment}

Next, we study entry deterrence when firms choose between models with different costs. To do so, we allow the entry cost to depend on the model choice. We find that the incumbent firm can choose a model that is too expensive to deter entry.

Suppose that the conditional biases $\mathbb{E}_{D}[\hat{f}_D(x) \mid f,x]-f(x)$ are equal for all models $\mathcal{M}_1$ and that the cost $c_m(M)$ is a decreasing and continuous function of the expected variance $$V=\mathbb{E}_{f,D,x}[{(\mathbb{E}_{D}[\hat{f}_D(x)\mid f, x]-\hat{f}_D(x))^2}].$$ One example is linear regression (which has zero bias), with different models corresponding to different dataset sizes. Collecting more data points is more costly but decreases variance.

When the set of models available to firm $1$ has the structure described in the previous paragraph, firm $1$ can choose an inefficiently costly model to deter entry:
\begin{proposition}\label{p:sequential_cost}
    Let $M_1^* \in \mathcal{M}_1$ be an optimal choice of model for a monopoly incumbent. Suppose there exists a model $M_1'\in \mathcal{M}_1$ with cost $c(M_1')$ greater than $c(M_1^*)$. Then there is an open subset of pairs $(c_f,\underline{u})$ such that at equilibrium, firm 1 chooses a model $M_1 \in \mathcal{M}_1$ with $c(M_1)>c(M_1^*)$ and firm 2 does not enter.
\end{proposition}

The proposition shows that incumbent firms can overinvest in modeling costs to maintain market power. As in \Cref{p:sequential_bias}, this happens when entry costs are intermediate and the outside option is low. If we interpret the model cost as the cost of collecting data, this overinvestment corresponds to practices often termed building a `data moat' in the context of artificial intelligence. The proposition illustrates potential inefficiencies associated with these practices.

An intuition is that investing to reduce variance makes entry less profitable for the competitor, and the incumbent will make such an investment if doing so will deter entry. As in \Cref{p:sequential_bias}, the result implies two inefficiencies. First, entry is inefficiently low at the equilibrium described in the result. Second, the firm spends an inefficiently high amount on data even assuming a monopoly: the firm purchases data with cost higher than its value to the consumer.

\medskip

Propositions \ref{p:sequential_bias} and \ref{p:sequential_cost} show how an incumbent can deter entry by an incumbent by choosing an inefficient model, giving rise to inefficiently concentrated market structures. This is in spite of the fact that our framework favors competition and entry by allowing the consumer to buy all models available to her.

\section{Conclusion}\label{s:conclusion}

This paper proposed a theoretical framework to study how firms that sell prediction models compete in the market. We showed that market outcomes can be expressed in terms of the bias-variance decompositions of the models that firms sell. This tractable characterization of market outcomes allowed us to study how market structure depends on the statistical properties of the models available to firms.    

Our analysis delivers several key insights. First, symmetric firms can choose distinct models that are complementary, which allows them to charge higher prices. 
Second, firms may choose inefficiently biased or costly models to deter entry by competitors or to achieve higher profits.

Our microfoundation for model choice has potential applications beyond the setup in this paper. Market outcomes would continue to depend on the bias-variance decomposition of the models that firms sell in frameworks that need not have all the features studied here (e.g., if consumers only purchase a single model). One natural direction is allowing heterogeneous consumers, which would give another reason for model differentiation. Another is considering objectives or welfare functions depending on more than mean-squared prediction error. As one example, if the agent we term the `consumer' is itself a firm predicting willingness-to-pay and setting prices, better predictions may not always be desirable.


\begin{singlespacing}

\bibliography{prediction}
\end{singlespacing}

\appendix

\section{Proofs} \label{proofs}
\subsection{Proof of \Cref{p:prices1}}
\begin{proof}
We start by showing that equation \eqref{e:prices} has a solution with non-negative prices. Fix $\mathbf M$, and let $\Psi(p): \mathbb R^{E(\mathbf M)} \rightarrow \mathbb R^{E(\mathbf M)}$ be defined by
$$\Psi_i(p) = \min_{E' \subset E(\mathbf M) \backslash \{i\}}  U( E(\mathbf{M}), \mathbf M) - U(E', \mathbf{M}) - \sum_{\substack{j \in E(\mathbf{M}) \backslash E' \\  j \neq i}} p_j$$
for each $i \in E(\mathbf M)$. Define 
\begin{align*}
\overline  p & \equiv \max_{j \in E(\mathbf M)  } U( E(\mathbf{M}), \mathbf M) - U(E(\mathbf{M} ) \backslash \{j\}, \mathbf{M})\\
\underline p & \equiv - (|E(\mathbf M)  -1|) \overline p .
\end{align*}
Note that $\Psi$ is continuous, and maps prices in $[\underline p,\overline  p]^{E(\mathbf M) }$ to prices in $[\underline p,\overline  p]^{E(\mathbf M) }$. Indeed, for every $p \in [\underline p,\overline  p]^{E(\mathbf M) }$ and every $i \in E(\mathbf M)$, 
$$\Psi_i(p) \leq U( E(\mathbf{M}), \mathbf M) - U(E(\mathbf{M} ) \backslash \{i\}, \mathbf{M}) \leq \ol p,$$
and 
$$\Psi_i(p) \geq - \sum_{j \neq i} p_j \geq \ul p,$$
where the last inequality follows from the inequality $p_j \leq \ol p$ and the definition of $\ul p$.
Hence, by Brouwer's fixed point theorem, there exists $p \in [\underline p,\overline p]^{E(\mathbf M) }$ with $\Psi(p) = p$. 

Next, we show that any fixed point $p$ of $\Psi$ satisfies $p \geq 0$. Let $p$ be a fixed point of $\Psi$. Fix $i \in E(\mathbf M)$, and let 
$$E \in \arg \min_{E' \subset E(\mathbf M) \backslash \{i\}}  U( E(\mathbf{M}), \mathbf M) - U(E', \mathbf{M}) - \sum_{\substack{j \in E(\mathbf{M}) \backslash E'\\  j \neq i}} p_j.$$
We consider two cases: (i) $|E(\mathbf M) \backslash E| = 1$, so $E = E(\mathbf M)\backslash \{i\}$, and (ii) $|E(\mathbf M) \backslash E| \geq 2$. In case (i), 
$$p_i =  U( E(\mathbf{M}), \mathbf M) - U(E, \mathbf{M}) - \sum_{\substack{j \in E(\mathbf{M}) \backslash E \\  j \neq i}} p_j = U( E(\mathbf{M}), \mathbf M) - U(E(\mathbf M)\backslash \{i\}, \mathbf{M})\geq 0,$$
since adding a model can only weakly increase the consumer's payoff (because the consumer can always assign zero weight to the added model).
In case (ii), there exists $k \neq i$ with $k \notin E$, and so
\begin{align*}
p_k & =  \min_{E' \subset E(\mathbf M) \backslash \{k\}}  U( E(\mathbf{M}), \mathbf M) - U(E', \mathbf{M}) - \sum_{\substack{j \in E(\mathbf{M}) \backslash E'\\  j \neq k}} p_j \\
& \leq U( E(\mathbf{M}), \mathbf M) - U(E \cup \{i\}, \mathbf{M}) - \sum_{\substack{j \in E(\mathbf{M}) \backslash E \cup \{i\} \\  j \neq k}} p_j\\
\Longrightarrow \sum_{j \in E(\mathbf{M}) \backslash E \cup \{i\}} p_j &\leq U( E(\mathbf{M}), \mathbf M) - U(E \cup \{i\}, \mathbf{M}).
\end{align*}
Hence, 
\begin{align*}
p_i &=  U( E(\mathbf{M}), \mathbf M) - U(E, \mathbf{M}) - \sum_{\substack{j \in E(\mathbf{M}) \backslash E \\  j \neq i}} p_j\\
& \geq U( E(\mathbf{M}), \mathbf M) - U(E, \mathbf{M}) - \left( U( E(\mathbf{M}), \mathbf M) - U(E \cup \{i\}, \mathbf{M}) \right) \\
& = U(E \cup \{i\}, \mathbf{M}) - U(E, \mathbf{M}) \geq 0.
\end{align*}   
Hence, if $p = \Psi(p)$, $p \geq 0$.  

Next, we show that given any $p$ solving \eqref{e:prices}, the subgame that starts at the pricing stage has a SPNE in which firms set prices $p$, and in which the consumer buys all models in $E(\mathbf M)$.   

Consider first the case with $N_E= |E(\mathbf M)| = 1$, and note that the price $p_i$ of the only entrant in \eqref{e:prices} is
$p_i = U(E(\mathbf M),\mathbf M) - U(\emptyset,\mathbf M) = U(E(\mathbf M),\mathbf M) - \ul u$: i.e., the entrant charges a price that extracts all the surplus from the consumer. Clearly, this is the only equilibrium price of the pricing subgame when there is one entrant. 

Next, consider the case with $N_E\geq 2$. Note first that, under the prices in \eqref{e:prices}, the consumer finds it optimal to purchase all models. To see this, note that
\begin{align*}
\forall i, \quad p_i &= \min_{E' \subset E(\mathbf M) \backslash \{i\}}  U( E(\mathbf{M}), \mathbf M) - U(E', \mathbf{M}) - \sum_{\substack{j \in E(\mathbf{M}) \backslash E' \\  j \neq i}} p_j \\
\Longleftrightarrow \forall i, \quad  U( E(\mathbf{M}), \mathbf M) - \sum_{j \in E(\mathbf M)} p_j &= \max_{E' \subset E(\mathbf M)\backslash \{i\}}  U(E', \mathbf{M}) - \sum_{j \in E'} p_j.
\end{align*}
Hence, at these prices, the consumer (weakly) prefers to buy all models in $E(\mathbf M)$ than to buy any subset of models (including not buying any model). 

Next, note that the payoff that the consumer obtains from purchasing all models is 
$$U(E (\mathbf M), \mathbf M ) - \sum_{j \in E(\mathbf M)} p_j,$$
while the payoff that the consumer gets from not purchasing model $i \in E(\mathbf M)$ is
$$\max_{E' \subset E(\mathbf M)\backslash\{i\}} U(E', \mathbf M ) - \sum_{j \in E'} p_j.$$
Hence, in any equilibrium of the pricing subgame in which the consumer buys all the models, firm $i$ charges a price $p_i$ that leaves the consumer indifferent between buying all models or not buying model $i$: 
\begin{align*}
    0 = &U(E (\mathbf M), \mathbf M ) - \sum_{j \in E(\mathbf M)} p_j - \left(\max_{E' \subset E(\mathbf M)\backslash\{i\}} U(E', \mathbf M ) - \sum_{j \in E'} p_j.\right)\\
    \Longleftrightarrow p_i  =& \min_{E' \subset E(\mathbf M) \backslash \{i\}}  U( E(\mathbf{M}), \mathbf M) - U(E', \mathbf{M}) - \sum_{j \in E(\mathbf{M}) \backslash E',  j \neq i} p_j.
\end{align*}
\end{proof}

\subsection{Proof of \Cref{cor:prices}}
\begin{proof}
Fix models $\mathbf M$, and suppose firm prices $(p_i)$ are given by \eqref{e:prices2}. To establish the result, we show that, under these prices, for each $i \in E(\mathbf M)$, 
\begin{equation}\label{e:max} 
\max_{E' \subset E(\mathbf M) \backslash \{i\}} U(E',\mathbf M) + \sum_{\substack{j \in E(\mathbf M)\backslash E'\\ j \neq i}} p_j =U(E(\mathbf M)\backslash \{i\},\mathbf M).
\end{equation}
Note that this implies that prices $(p_i)$ given by \eqref{e:prices2} satisfy \eqref{e:prices}. Proposition \ref{p:prices1} then implies that these prices form a SPNE of the pricing subgame. 

Note first that, when $|E(\mathbf M)| = 1$, equation \eqref{e:max} automatically holds.  
Consider next the case with $|E(\mathbf M)| \geq 2$. Pick $i \in E(\mathbf M)$, and $E' \subset E(\mathbf M) \backslash \{i\}, E' \neq E(\mathbf M) \backslash \{i\}$, so that there exists $k \neq i$ with $k \in E(\mathbf M)$ but $k \notin E'$. Let $E'' =E' \cup \{k\}$, and note that
\begin{align*} 
&\left(U(E',\mathbf M) + \sum_{\substack{j \in E(\mathbf M)\backslash E'\\ j \neq i}} p_j\right) - \left(U(E'',\mathbf M) + \sum_{\substack{j \in E(\mathbf M)\backslash E'' \\ j \neq i}} p_j\right) \\
=& U(E'' \backslash \{k\},\mathbf M) - U(E'',\mathbf M) + p_k \\
=& U(E'' \backslash \{k\},\mathbf M) - U(E'',\mathbf M) + U(E(\mathbf M),\mathbf M) - U(E(\mathbf M)\backslash \{k\},\mathbf M)\leq 0,
\end{align*}
where the second equality follows since prices satisfy \eqref{e:prices2}, and the inequality follows from decreasing marginal returns. Hence, 
$$U(E'',\mathbf M) + \sum_{\substack{j \in E(\mathbf M)\backslash E''\\ j \neq i}} p_j = U(E'\cup \{k\},\mathbf M) + \sum_{\substack{j \in E(\mathbf M)\backslash E'\cup \{k\}\\ j \neq i}} p_j \geq U(E',\mathbf M) + \sum_{\substack{j \in E(\mathbf M)\backslash E'\\ j \neq i}} p_j.$$
Since this inequality holds for all $E' \subset E(\mathbf M) \backslash \{i\}, E' \neq E(\mathbf M) \backslash \{i\}$, it follows that
$$\forall E' \subset E(\mathbf M)\backslash\{i\}, \quad U(E(\mathbf M)\backslash\{i\},\mathbf M) \geq U(E',\mathbf M) + \sum_{\substack{j \in E(\mathbf M)\backslash E'\\ j \neq i}} p_j,$$
and so \eqref{e:max} holds. 
\end{proof}

\subsection{Proof of \Cref{p:efficiency}}

\begin{proof}
Fix models $\mathbf{M} = (M_i)$, and for each $i$ let $\hat f_{i,D^{(i)}}(x) =M_i(x,D^{(i)})$. The total surplus from models $\mathbf{M}$ is
$$TS(\mathbf{M}) = -\esp\left[\left( \sum_{i\in E(\mathbf{M})}w_i^{E(\mathbf{M})} \hat f_{i,D^{(i)}}(x) - y\right)^2 \right] - \sum_{i \in E(\mathbf{M})} c(M_i) - \ul u.$$

Suppose the action profile $\mathbf{M}^*$ maximizes total surplus. For each firm $i$, let $\hat f^*_{i,D^{(i)}}(\cdot) = M^*_i(\cdot, D^{(i)})$.
Consider firm $j$ and model $M_j' \neq M^*_j$. Let $\hat f'_{j,D^{(j)}}(\cdot) = 
M_j'(\cdot, D^{(j)})$. Given optimal weights and the prices in Corollary 
\ref{p:prices1}, the profit (net of entry costs) that firm $j$ gets under action profile $\mathbf{M}' = 
\left(\mathbf{M}^*_{-j},M_j'\right)$ is
\begin{align*}
\Pi_j(\mathbf{M}') =& -\esp\left[\left( \sum_{i\in E(\mathbf{M}')} w_i^{E(\mathbf{M}')}\hat f_{i,D^{i}}(x)  - y\right)^2 \right] \\
& \quad +\esp\left[\left( \sum_{i\in E(\mathbf{M}^*)\backslash \{j\}} w_i^{E(\mathbf{M}^*)\backslash\{j\}} \hat f^*_{i,D^{i}}(x) - y\right)^2 \right] - c(M_j') \\
=& TS(\mathbf{M}') - TS\left(\mathbf{M}^*_{-j},\emptyset\right). 
\end{align*}
Note then that, 
\begin{align*}
    \Pi_j(\mathbf{M}') - \Pi_j(\mathbf{M}^*) =& \quad TS(\mathbf{M}') - TS\left(\mathbf{M}^*_{-j},\emptyset\right) - TS(\mathbf{M}^*) + TS\left(\mathbf{M}^*_{-j},\emptyset\right) \\
    =& \quad TS(\mathbf{M}') - TS(\mathbf{M}^*) \\
    \leq & \quad  0.
\end{align*}
Hence, the game has a SPNE in which firms choose $\mathbf{M}^*$ and set the corresponding prices in Corollary \ref{p:prices1}, and the consumer purchases all models and chooses weights optimally. 
\end{proof}

\subsection{Proof of \Cref{p:symmetric_deterministic}}

We begin with a lemma characterizing when models satisfy decreasing marginal returns.

\begin{lemma}\label{lem:symmetric}
Models satisfy decreasing marginal returns when all firms have access to the same model and $\ul u< - B - \frac{3}{2}V - \sigma^2$. 
\end{lemma}

\begin{proof}
Note first that, when all firms choose the same model, it is optimal for the consumer to assign equal weights to each model. Hence, for any set $E \subset E(\mathbf M)$ 
\begin{align*}
U(E,\mathbf M) & = - \esp_{f,D,x}\left[\left( \sum_{i\in E}\frac{1}{|E|}\hat f_{i,D^i} - y \right)^2 \right]\\
& = - B - \frac{1}{|E|}V - \sigma^2,
\end{align*}
where the last equality uses Lemma \ref{p:bias_variance}, together with all models being identical with bias $B$ and variance $V$, and with datasets being iid across firms (so that the variance of the averaged prediction is $1/|E|$ the variance of each model). 

Fix $E \subset E(\mathbf M)$, a nonempty $E' \subset E$, and $j \in E'$. If $|E'| \geq 2$, then 
\begin{align*}
    U(E',\mathbf M) - U(E'\backslash\{j\},\mathbf M) &= \frac{V}{|E'|(|E'| - 1)}  \\
    &\geq \frac{V}{|E|(|E| - 1)} = U(E,\mathbf M) - U(E\backslash\{j\},\mathbf M).
\end{align*}
If $|E'| = 1$, then $E' = \{j\}$, and so
$$ U(E',\mathbf M) - U(E'\backslash\{j\},\mathbf M) = - B - V - \sigma^2- \ul u.$$
We consider two cases when $|E'|=1$. If $E' =E$, then $$U(E',\mathbf M) - U(E'\backslash\{j\},\mathbf M) = U(E,\mathbf M) - U(E\backslash\{j\},\mathbf M).$$
If $E' \subsetneq E$, then
\begin{align*}
    U(E',\mathbf M) - U(E'\backslash\{j\},\mathbf M) &= - B - V - \sigma^2- \ul u  \\
    &\geq \frac{1}{2} V\geq \frac{V}{|E|(|E| - 1)} = U(E,\mathbf M) - U(E\backslash\{j\},\mathbf M),
\end{align*}
where the first inequality follows by assumption.
\end{proof}

The proposition follows from a second lemma, which describes equilibrium outcomes and surpluses:

\begin{lemma}\label{l:symmetric_deterministic}
Suppose all firms are symmetric and have access to the same model $M$ with expected variance $V$ and expected squared bias $B$. 

In every pure-strategy SPNE, firms' entry decisions maximize total surplus. 
\begin{enumerate}
\item[(i)] If $V<2 c$, in every pure-strategy SPNE one firm enters, and sets price $p = -B-V-\sigma^2-\underline{u}$. The consumer surplus is 0.
\item[(ii)] If $V\geq 2 c$, in every pure-strategy SPNE the number of entrants is 
\begin{equation}\label{e:entry}
    N_E = \max\left\{j \in \NN : \frac{V}{j(j-1)}  \geq c \right\},\end{equation}
and all firms that enter charge price 
$p = \frac{V}{N_E(N_E-1)}.$ The consumer surplus is $ - B - \frac{2N_E - 1}{N_E(N_E - 1)}V - \sigma^2 - \ul u$. 
\end{enumerate}
\end{lemma}

\begin{proof}
Consider first the pricing subgame. Suppose that $N_E=|E(\mathbf{M})|$ firms entered the market, and each choose the model $M$. Since all firms have the same model, by Corollary \ref{cor:prices}, in an equilibrium in which all models are bought, all firms set a price equal to their marginal contribution. Further, since each firm sells the same model given their data, the consumer must weigh each model equally, that is, $w_i=\frac{1}{N_E}$ for all $i$. 



Using the bias-variance decomposition in Lemma \ref{p:bias_variance}, we can explicitly compute the equilibrium price. If there is only one entrant, 
\[
p= -\mathbb{E}_{f,x}\left[(\hat{f}_{1,D^{(1)}}(x)-y)^2\right] -\underline{u} = -B-V-\sigma^2-\underline{u}.
\]

If there are two or more entrants, $N_E\geq 2$, then
\begin{align*}
    p=&  -\mathbb{E}_{f,D,x}\left[ \left( \frac{1}{N_E} \sum_{i=1}^{N_E} \hat{f}_{i,D^i}(x)-y\right)^2\right]+\mathbb{E}_{f,D,x}\left[ \left( \frac{1}{N_E-1} \sum_{i=1}^{N_E-1} \hat{f}_{i,D^i}(x)-y\right)^2\right] \\
    =& \quad \mathbb{E}_{f,D,x} \left[ -\frac{1}{N_E^2} \sum_{i=1}^{N_E} \mathrm{Var}_D\left(\hat{f}_{i,D^i}(x)\mid x\right) + \frac{1}{(N_E-1)^2} \sum_{i=1}^{N_E-1} \mathrm{Var}_D\left(\hat{f}_{i,D^i}(x)\mid x\right)\right] \\
    =& \quad \frac{V}{N_E(N_E-1)}.
\end{align*}
Note that price is strictly decreasing in the number of firms in the market when $N_E\geq 2$, and that a firm's profit is independent of $B$. Intuitively, the marginal value for the consumer of purchasing one model is that it reduces the variance. At the same time, since all firms have access to the same model, buying one more model leaves the bias unchanged.  

Next, consider the stage where the firms are making entry decisions. The equilibrium price $p$ should be higher than the cost of entry $c$ so that firms are willing to enter. 
Hence, the number of entrants is given by $N_E$ in \eqref{e:entry}. One can check that $N_E$ maximizes total surplus $TS = - B - \frac{1}{N_E}V - \sigma^2- N_E c -\ul u$. 


Now we consider surpluses.
A firm's surplus in equilibrium is just the price minus the entry cost:
\[
\begin{cases}
\frac{V}{N_E(N_E-1)}-c & \text{ if } N_E>1 \\
-B-V-\sigma^2-\underline{u}-c & \text{ if } N_E=1
\end{cases}
\]

For consumer surplus, if $N_E=1$, the market is a monopoly and equilibrium price is such that the consumer is indifferent between purchasing the model and the outside option. The monopolist extracts all surplus and thus, consumer surplus is 0. If $N_E\geq 2$, consumer surplus is equal to
\begin{align*}
& -\mathbb{E}_{f,x}\left[ \left( \frac{1}{N_E} \sum_{i=1}^{N_E} \hat{f}_i(x)-y\right)^2\right]-pN_E -\ul u \\
=& -B-\frac{V}{N_E}-\sigma^2-\frac{V}{N_E-1} - \ul u\\
=&  - B - \frac{2N_E - 1}{N_E(N_E - 1)}V - \sigma^2 - \ul u.
\end{align*}
This completes the proof.
\end{proof}
\subsection{Proof of \Cref{prop:diff}}
\begin{proof}

Suppose firm $1$ chooses $M(t)$ and firm $2$ chooses $M(t_0).$ Let $w(t)$ be the weight the consumer places on firm $1$'s prediction. Firm $1$'s payoff is
\begin{align*}
\Pi(t) & = V(t_0)+B(t_0)- (1-w(t))^2(V(t_0)+B(t_0))- w(t)^2 (V(t) +B(t)) \\ & \quad- 2w(t)(1-w(t))\int b_0(f,x)b_t(f,x) dFdG- c
\\ & =V(t_0)+B(t_0)- (1-w(t))^2(V(t_0)+B(t_0))- w(t)^2 (V(t) +B(t)) \\ & \quad- 2w(t)(1-w(t))\sqrt{B(t_0)\,B(t)}\,\cos(\theta(t))- c.
\end{align*}
The result will follow from computing the first-order condition and second-order condition for optimality of the best response $M(t_0)$.

The optimal weights for the consumer also maximize firm $1$'s payoff. So by the envelope formula, we can calculate $\frac{d\Pi}{dt}$ taking $w(t)$ to be constant. We then compute

\begin{align*}
\frac{d\Pi}{dt}
  &=  -w(t)^2\,V'(t) -w(t)^2 \,B'(t)
     - 2w(t)(1-w(t))\sqrt{B(t_0)}\,\frac{d}{dt}\!\left[\sqrt{B(t)}\cos(\theta(t))\right] \\[4pt]
  &= -w(t)^2\,V'(t) -w(t)^2 \,B'(t)
     - 2w(t)(1-w(t))\sqrt{B(t_0)}\left(
        \frac{B'(t)}{2\sqrt{B(t)}}\cos(\theta(t))-
        \sqrt{B(t)}\sin(\theta(t))\,\theta'(t)
     \right).
\end{align*}

We now evaluate at $t=t_0$. By symmetry, we have $w(t_0) = \frac12$. Since $\theta(t_0) = 0$, we obtain the first-order condition $$-\frac14 V'(t_0) - \frac12 B'(t_0) = 0.$$

We next consider the second-order condition. To calculate the second derivative, it is helpful to write
$$ Z(t) = \sqrt{B(t_0)}\left(
        \frac{B'(t)}{2\sqrt{B(t)}}\cos(\theta(t))-
        \sqrt{B(t)}\sin(\theta(t))\,\theta'(t)
     \right) $$
and calculate 
\begin{align*}Z'(t) & = \sqrt{B(t_0)} \biggl(
        \frac{B''(t)}{2\sqrt{B(t)}}\cos(\theta(t))- 
        \frac{B'(t)^2}{4 B(t)^{3/2}}\cos(\theta(t))-
        \frac{B'(t)}{\sqrt{B(t)}}\sin(\theta(t)) \theta'(t)  \\ & \quad  - \sqrt{B(t)}\cos(\theta(t))\,\theta'(t)^2 - \sqrt{B(t)}\sin(\theta(t))\,\theta''(t)
     \biggr) .\end{align*}
We then have 
\begin{align*}
\frac{d^{2}\Pi}{dt^{2}}
  =&-w(t)^2\,V''(t) - w(t)^2   B''(t) - 2w(t)(1-w(t)) Z'(t)  \\ & \quad - 2w(t)w'(t) (V'(t) + B'(t)) -2 (1-2w(t))w'(t) Z(t).
\end{align*}

We want to evaluate this expression at $t=t_0$. Recall that $w(t_0)=\frac12$ and $\theta(t_0)=0$. We can compute $Z(t_0) = \frac{B'(t_0)}{2}$ and $Z'(t_0) = \frac{ B''(t_0)}{2} - \frac{B'(t_0)^2}{4 B(t_0)} -B(t_0) \theta'(t_0)^2$. So 
\[
{
\left.\frac{d^{2}\Pi}{dt^{2}}\right|_{t=t_0}
  = -\frac14\,V''(t_0)
    -\frac12\,B''(t_0)
    - w'(t_0) (V'(t_0) + B'(t_0)) 
    +\frac{B'(t_0)^{2}}{8\,B(t_0)}
    +\frac{B(t_0)}{2}\theta'(t_0)^{2}
}.
\]
The second-order condition requires that this quantity is non-positive. To complete the proof, we will show that $$-w'(t_0)(V'(t_0) + B'(t_0)) = \frac{B'(t_0)^2}{4V(t_0)}.$$

We want to find $w'(t_0)$. The weight $w(t)$ is chosen to minimize
$$w(t)^2 (V(t) + B(t)) + (1-w(t))^2(V(t_0) + B(t_0))  + 2w(t)(1-w(t)) \sqrt{B(t)B(t_0)}\cos(\theta(t)).$$
The first-order condition for $w(t)$ gives
$$w(t)(V(t) +B(t)) + (w(t)-1)(V(t_0)+B(t_0)) + (1-2w(t)) \sqrt{B(t)B(t_0)}\cos(\theta(t))=0 .$$
Implicitly differentating in $t$, we obtain

\begin{align*}w'(t)(V(t) + B(t) + V(t_0)+B(t_0) - 2\sqrt{B(t)B(t_0)}\cos(\theta(t)))  \\ + w(t)(V'(t) + B'(t)) + (1- 2w(t)) \sqrt{B(t_0)}\,\frac{d}{dt}\!\left[\sqrt{B(t)}\cos(\theta(t))\right] & = 0. \end{align*}

Substituting $t=t_0$,  we have 
$$w'(t_0) = -\frac{V'(t_0) + B'(t_0)}{4V(t_0)}.$$
Finally, the first-order condition $V'(t_0) + 2B'(t_0)=0 $
gives
$$-(V'(t_0) + B'(t_0))w'(t_0) = \frac{B'(t_0)^2}{4V(t_0)}.$$
Substituting into the expression for $\left.\frac{d^{2}\Pi}{dt^{2}}\right|_{t=t_0}$ above gives the result.
\end{proof}

\subsection{Differentiation in Linear Regression Models}\label{a:diff_lr}

Let $\ol \beta^2 = \esp_{\beta}[\beta_l^2]$. The following Lemma gives the bias-variance decomposition of the expected mean squared
error for the linear regression setting. 

\begin{lemma}\label{l:ols} Suppose firm $i$ chooses model $M_i$ that includes $|M_i|\leq k$ covariates. Then, firm $i$'s prediction error is 
\begin{align}
\esp[(y - x^T \hat \beta^{(i)})^2] &= \sigma^2 + \underbrace{\ol \beta^2(k - |M_i|) }_{B_i = \esp [(x^T \beta - x^T \esp[\hat \beta^{(i)}] )^2] } +  \underbrace{(\ol \beta^2(k - |M_i|) + \sigma^2)\left(\frac{|M_i|}{n-|M_i|-1}\right)}_{V_i = \esp [(x^T \esp[\hat \beta^{(i)}] - x^T \hat \beta^{(i)})^2]}. \label{eq:error_lr}
\end{align}
\end{lemma}
\begin{proof}
Fix a model $M_i$ with $|M_i| = d\leq k$, and assume wlog that model $M_i$ includes covariates $l = 1,...,d$.
By Lemma \ref{p:bias_variance}, we have the following bias-variance decomposition:
\begin{equation}\label{eq:decomposition}
\esp[(y - x^T \hat \beta^{(i)})^2] = \sigma^2 + \esp_{\beta,x} [(x^T \beta - x^T \esp_{D^{(i)} | \beta}[\hat \beta^{(i)}] )^2] +  \esp_{\beta,D^{(i)},x} [(x^T \esp_{D^{(i)}| \beta}[ \hat \beta^{(i)}] - x^T \hat \beta^{(i)})^2].
\end{equation}
Consider the second term on the RHS of \eqref{eq:decomposition}, and note that 
\begin{align*}
\esp_{\beta,x} [(x^T \beta - x^T \esp_{D^{(i)}| \beta}[\hat \beta^{(i)}])^2] & = \esp_{\beta,x} \left[\left(\sum_{l=1}^k x_l(\beta_l - \esp_{D^{(i)}| \beta}[\hat \beta^{(i)}_{l}])\right)^2\right]\\
&=   \esp_{\beta,x}\left[\sum_{l=1}^k x^2_l(\beta_l - \esp_{D^{(i)}| \beta}[\hat \beta^{(i)}_{l}])^2\right]\\
& =  \left[\sum_{l=1}^k \esp_x[x^2_l]\esp_{\beta}(\beta_l - \esp_{D^{(i)}| \beta}[\hat \beta^{(i)}_{l}])^2\right]  = \esp_{\beta}[ \|\beta - \esp_{D^{(i)}| \beta}[\hat \beta^{(i)}]\|^2],
\end{align*}
where the second and third equalities use $x \sim N(0,I_k)$ and $x$ independent of $D^{(i)} = (X_i,Y_i)$ and $\beta$, and the last equality again uses $x \sim N(0,I_k)$, and so $\esp[x^2_l] = 1$ for all $l$.

Consider next the last term on the RHS of \eqref{eq:decomposition}: 
\begin{align*}
\esp_{\beta,D^{(i)},x} [(x^T \esp_{D^{(i)}| \beta}[ \hat \beta^{(i)}] - x^T \hat \beta^{(i)})^2] &= \esp_{\beta,D^{(i)},x}\left[\left(\sum_{l=1}^k x_l\left(\esp_{D^{(i)}| \beta}[\hat \beta^{(i)}_{l}]- \hat \beta^{(i)}_{l}\right) \right)^2\right]\\
&=   \esp_{\beta,D^{(i)},x}\left[\sum_{l=1}^k x^2_l\left(\esp_{D^{(i)}| \beta}[\hat \beta^{(i)}_{l}] - \hat \beta^{(i)}_{l}\right)^2\right]\\
& =  \left[\sum_{l=1}^k \esp_x[x^2_l]\esp_{\beta,D^{(i)}}\left(\esp_{D^{(i)}| \beta}[\hat \beta^{(i)}_{l}] - \hat \beta^{(i)}_{l} \right)^2\right] \\
&= \esp_{\beta,D^{(i)}}[ \|\esp_{D^{(i)}| \beta}[\hat \beta^{(i)}]- \hat \beta^{(i)} \|^2].
\end{align*}
Hence, 
$$\esp[(y - x^T \hat \beta^{(i)})^2] = \sigma^2 + \esp_\beta[ \|\beta - \esp_{D^{(i)}| \beta}[\hat \beta^{(i)}]\|^2] + \esp_{\beta,D^{(i)}}[ \|\esp_{D^{(i)}| \beta}[\hat \beta^{(i)}]- \hat \beta^{(i)} \|^2]$$

Let $\beta_{M_i} = (\beta_l)_{l \in M_i}$ denote the coefficients of the covariates included in $M_i$. Note that 
$$\hat \beta_{M_i}^{(i)} = (X^T_{i,M_i} X_{i,M_i})^{-1}X^T_{i,M_i} Y_i = (X^T_{i,M_i} X_{i,M_i})^{-1}X^T_{i,M_i}(X_{i,M_i} \beta_{M_i} + \eta ) =\beta_{M_i} +(X^T_{i,M_i} X_{i,M_i})^{-1}X^T_{i,M_i} \eta$$
where $\eta = Y_i - X_{i,M_i} \beta_{M_i} = \sum_{l=d+1}^k x_{i,l} \beta_l + \epsilon \sim N(0,\sigma_d^2 \times I)$ with $\sigma^2_d = \sum_{l=d+1}^k \beta^2_l + \sigma^2$.  Then, 
\begin{align*}
\esp_{D^{(i)}|\beta }[\hat \beta_{M_i}^{(i)}] = \beta_{M_i} + \esp_{D^{(i)}| \beta}[(X^T_{i,M_i} X_{i,M_i})^{-1}X^T_{i,M_i} \eta] =\beta_{M_i},  
\end{align*}
where we used  $\esp_{D^{(i)}| \beta}[(X^T_{i,M_i} X_{i,,M_i})^{-1}X_{i,M_i}^T\eta] = 0$ (since $\eta$ is independent of $X_{i,M_i}$, and $\esp_{D^{(i)}| \beta}[\eta]=0$). Since $\hat \beta^{(i)}_{l} = 0$ for all $l \notin M_i$, and using $\beta_{M_i^c} = (\beta_l)_{l \notin M_i}$ and $\ol \beta^2 = \esp_\beta[\beta_l^2]$, we get
$$\esp_\beta[ \|\beta - \esp_{D^{(i)}| \beta}[\hat \beta^{(i)}]\|^2] =  \esp_\beta [\|\beta_{M_i^c}\|^2] = \overline \beta^2(k -d) = \overline \beta^2(k -|M_i|).$$

Consider next the term $ \esp_{\beta,D^{(i)}}[ \|\esp_{D^{(i)}| \beta}[\hat \beta^{(i)}]- \hat \beta^{(i)} \|^2]$. Note that, since $\hat \beta^{(i)}_{l} = 0$ for all $l \notin M_i$, $ \esp_{\beta,D^{(i)}}[ \|\esp_{D^{(i)}| \beta}[\hat \beta^{(i)}]- \hat \beta^{(i)} \|^2] =  \esp_{\beta,D^{(i)}}[ \|\esp_{D^{(i)}| \beta}[\hat \beta_{M_i}^{(i)}]- \hat \beta_{M_i}^{(i)} \|^2]$. Hence, using $\hat \beta_{M_i}^{(i)} = \beta_{M_i} +(X^T_{i,M_i} X_{i,M_i})^{-1}X^T_{i,M_i} \eta $ and $\esp_{D^{(i)}| \beta}[\hat \beta^{(i)}_{M_i}] = \beta_{M_i}$, we get
\begin{align*}
\esp_{\beta,D^{(i)}}[ \|\esp_{D^{(i)}| \beta}[\hat \beta^{(i)}]- \hat \beta^{(i)} \|^2] & = \esp_{\beta} [\esp_{D^{(i)}|\beta} [ \|(X^T_{i,M_i} X_{i,M_i})^{-1}X^T_{i,M_i} \eta  \|^2]]\\
& =\esp_\beta[ \esp_{D^{(i)}| \beta}[\eta^T X_{i,M_i}(X^T_{i,M_i} X_{i,M_i})^{-1}(X^T_{i,M_i} X_{i,M_i})^{-1}X^T_{i,M_i} \eta ]]\\
& =\esp_\beta[ \esp_{D^{(i)}| \beta}[tr((X^T_{i,M_i} X_{i,M_i})^{-1}(X^T_{i,M_i} X_{i,M_i})^{-1}X^T_{i,M_i} \eta\eta^T X_{i,M_i}) ]]\\
& =tr(\esp_\beta[ \esp_{D^{(i)}| \beta}[(X^T_{i,M_i} X_{i,M_i})^{-1}(X^T_{i,M_i} X_{i,M_i})^{-1}X^T_{i,M_i} \esp_\eta[\eta\eta^T] X_{i,M_i} ]])\\
& =  tr\left(\esp_\beta\left[\esp_{D^{(i)}| \beta}[ (X^T_{i,M_i} X_{i,M_i})^{-1} ] \sigma_d^2 \right] \right)\\
& = \frac{d}{n-1-d}(\overline \beta^2(k - d) + \sigma^2) 
\end{align*}
where the fourth equality uses the independence of $\eta$ and  $X_{i,M_i}$, the fifth equality uses $\esp_{\eta}[\eta\eta^T] = \sigma^2_d I$ (since $\eta \sim N(0,\sigma_d^2 I)$, with $\sigma_d^2 = (\sum_{l=d+1}^k \beta^2_l + \sigma^2 )$), and the last equality follows since $\esp_\beta[\beta_l^2] = \ol \beta^2$ and since $(X^T_{i,M_i} X_{i,M_i})^{-1}$ has an inverse-Wishart distribution with $n$ degrees of freedom and scale matrix $I_d$, and so  $\esp_{D^{(i)}|\beta }[ (X^T_{i,,M_i} X_{i,M_i})^{-1} ] =  \frac{1}{n-1-d}I_d$.
\end{proof}

Suppose firms 1 and 2 enter the market with models $M_1$ and $M_2$. Let $w \in[0,1]$ denote the weight that the consumer puts on model 1. The price the consumer pays for model 1 is
\begin{align}
p_1 & =  \esp[(y - x^T \hat \beta^{(2)})^2] - \esp[(y - w x^T \hat \beta^{(1)} - (1-w) x^T \hat \beta^{(2)})^2] \notag \\
& = \esp[(y - x^T \hat \beta^{(2)})^2] - \sigma^2 - w^2 B_1 - w^2 V_1 - (1 - w)^2 B_2 - (1 - w)^2 V_2 \notag \\
&\quad -2 w (1-w) \esp_{f,x}[(f(x) - \esp_{D^{(1)}}[\hat f_{1,D^{(1)}}(x)  | f,x])(f(x) - \esp_{D^{(2)}}[\hat f_{2,D^{(2)}}(x)  | f,x])]\notag \\
& = \esp[(y - x^T \hat \beta^{(2)})^2]- \sigma^2 - w^2 B_1 - w^2 V_1 - (1 - w)^2 B_2 - (1 - w)^2 V_2 \notag \\
& \quad- 2 w(1-w) \ol \beta^2 |\{ l \in \{1,..,k\} : l \notin M_1 \cup M_2 \}|, \label{eq:price_lr}
\end{align}
where $\hat f_{i,D^{(i)}}(x)  =  x^T \hat \beta^{(i)}$ is $i$'s prediction. The last equality follows since, for $i=1,2$, $$f(x) - \esp_{D^{(i)}}[\hat f_{i,D^{(i)}}(x)  | f,x] = x^T \beta - x^T \esp_{D^{(i)}|\beta }[ \hat \beta^{(i)}] = \sum_{l \notin M_i}\beta_l x_l,$$ and so
\begin{align*}
 \quad \quad & \; \esp_{f,x}[(f(x) - \esp_{D^{(1)}}[\hat f_{1,D^{(1)}}(x)  | f,x])(f(x) - \esp_{D^{(2)}}[\hat f_{2,D^{(2)}} (x) | f,x])] \\
=& \; \esp_{f,x}\left[\sum_{l\notin M_1}(\beta_l x_l) \sum_{l' \notin M_2}(\beta_{l'} x_{l'}) \right]\\
=& \; \esp_{f,x}\left[\sum_{l \notin M_1 \cup M_2} x_l^2 \beta_l^2\right] \\
=& \; \ol \beta^2 |\{ l \in \{1,..,k\} : l \notin M_1 \cup M_2 \}|,
\end{align*}
where the second and third equalities use $x_l \sim N(0,I_k)$ and $\beta$ independent of $x$. 

\medskip

Assume that the entry cost $c(|M|)$ is strictly increasing and convex, with $c(1)$ small. We claim that if $k$ is large enough, there is no equilibrium in which both firms enter with the same model. To see why, note first that since firms' payoff from entering is bounded above by the consumer's outside option, for $k$ large enough a firm that enters will choose a model that excludes some covariates. 

Suppose by contradiction that there exists an equilibrium in which both firms enter with the same model $M$. By our arguments above $M$ must exclude at least one covariate.  
Pick $l \in M$ and $l' \notin M$, and let $M' = M \cup \{l'\} \backslash \{l\}$. That is, model $M'$ is equal to $M$, except that it includes covariate $l'$ and does not include covariate $l$.  Let $B$ and $V$ denote respectively, expected squared bias and expected variance of model $M$. Note that since covariates are exchangeable, model $M'$ also has squared bias $B$ and variance $V$.

Firm 1's payoff from choosing model $M$ when firm 2 chooses $M$ is 
$$\esp[(y - x^T \hat \beta^{(2)})^2] - \frac{1}{2} (B +  V) - \frac{1}{2} \ol \beta^2 (k - |M|)-\sigma^2-c(|M|),$$
where we used  \eqref{eq:price_lr} together with the fact that the consumer optimally chooses a weight of 1/2 when both firms sell the same model. 

Instead, firm 1's payoff from choosing model $M'$ when firm 2 chooses $M$ is 
$$\esp[(y - x^T \hat \beta^{(2)})^2] - \frac{1}{2} (B +  V) - \frac{1}{2} \ol \beta^2 (k - |M| - 1)-\sigma^2-c(|M'|),$$
where we again used  \eqref{eq:price_lr} and the fact that the consumer optimally chooses a weight of 1/2 when firms choose models $M$ and $M'$. Since models $M$ and $M'$ are equally costly, choosing model $M'$ is a strictly profitable deviation. 

\medskip

Next, we consider entry costs that do not depend on model size, i.e., $c(|M|) = c$ for all $M$. By \Cref{l:ols}, the model with covariates $\{1,\hdots,d\}$ for $0 \leq d \leq k$ has expected squared bias and expected variance
$$B_d = 
\ol \beta^2(k - d) \text{ and } V_d=(\ol \beta^2(k - d) + \sigma^2)\left(\frac{d}{n-d-1}\right). 
$$
The consumer's utility from prediction if both firms use this model is 
$$U(d) =- \sigma^2 - B_d - \frac12 V_d.$$
We compute that
$$U'(0) =- \frac{\ol \beta^2 (k-2n+2) + \sigma^2}{2n-2} \text{ and }U'(k) = \frac{(\ol \beta^2 (k^2 -3 k (n-1) + 2 (n-1)^2)) - ( n-1) \sigma^2}{2 (n-k-1)^2}.$$
We have $U'(0)>0$ and $U'(k)<0$ if
$$\frac{k^2}{n-1} - 3k + 2n-2 < \frac{\sigma^2}{\ol \beta^2} <-k+ 2n  - 2.$$
We can always choose $\sigma^2$ so these inequalities are both satisfied. This implies $U(d)$ has an interior maximum on the interval $[0,k] \subset \mathbb{R}.$ We want to find parameters such that $U(d)$ has an interior maximum on $\{0,1,\hdots,k\} \subset \mathbb{Z}$.

Note the signs of $U'(0)$ and $U'(d)$ are unchanged if we rescale $k$, $n$, and $\sigma^2$ proportionally, and their values converge. So we can choose these variables sufficiently large so that $U(d)$ has an interior maximum on $\{1,\hdots,k\} \subset \mathbb{Z}$. For these parameters, the efficient set of covariates when both firms are constrained to use the same covariates is interior. Then the same argument as in the costly covariates case above shows that this cannot be an equilibrium, so there must be an equilibrium in which the firms use different covariates.

\subsection{Proof of Proposition \ref{p:differentiation_surplus}}

\begin{proof}
    The consumer's expected payoff at the equilibrium with models $(M_1,M_2)$ is
    \begin{align*}U(\{1,2\}, (M_1,M_2)) - p_1 - p_2 &= U(\{1,2\}, (M_1,M_2))- (U(\{1,2\}, (M_1,M_2)) - U(\{1\},(M_1,M_2))) \\ & \quad  -(U(\{1,2\}, (M_1,M_2)) - U(\{2\},(M_1,M_2))) \\ & = U(\{1\},(M_1,M_2))+U(\{2\},(M_1,M_2))-U(\{1,2\}, (M_1,M_2)) .\end{align*}

    We can assume without loss of generality that \begin{equation}\label{eq:model_ranking} U(\{1\},(M_1,M_2)) -c(M_1) \geq U(\{2\},(M_1,M_2))-c(M_2).\end{equation} The consumer's expected payoff in the pricing subgame after firms choose models $(M_1,M_1)$ is
\begin{align*}U(\{1,2\}, (M_1,M_1)) - p_1 - p_2 &= U(\{1,2\}, (M_1,M_1))- 2(U(\{1,2\}, (M_1,M_1)) - U(\{1\},(M_1,M_1))) \\ &  = 2U(\{1\},(M_1,M_1))-U(\{1,2\}, (M_1,M_1)) .\end{align*}

We must have \begin{equation}\label{eq:deviation}U(\{1,2\}, (M_1,M_1)) - c(M_1) \leq U(\{1,2\}, (M_1,M_2)) -c(M_2)\end{equation} 
because otherwise firm $2$ would deviate to choose $M_1$ in the equilibrium with models $(M_1,M_2)$. Further note that $U(\{1\},(M_1,M_2))=U(\{1\},(M_1,M_1))$. Combining inequalities  \eqref{eq:model_ranking} and \eqref{eq:deviation} and adding $U(\{1\},(M_1,M_1))$ to both sides gives $$  U(\{1\},(M_1,M_2))+U(\{2\},(M_1,M_2))-U(\{1,2\}, (M_1,M_2)) \leq 2U(\{1\},(M_1,M_1))-U(\{1,2\}, (M_1,M_1)).$$
This inequality shows that the consumer surplus at the equilibrium with differentiation is weakly less than consumer surplus in the equilibirum of the pricing subgame after both firms choose model $M_1$. We further note the inequality is strict if either of \eqref{eq:model_ranking} and \eqref{eq:deviation} are.
\end{proof}

\subsection{Proof of \Cref{p:sequential_bias}}
\begin{proof}
Let $V_1 = \mathbb{E}_{f,x,D}[(\mathbb{E}_{D|f}[\hat{f}_1(x)]-\hat{f}_1(x))^2]$ be the expected variance of $M_1$. The payoff to firm 2 from entering and choosing $M_2$ with bias constant $\alpha_2$ and expected variance $V_2 =\mathbb{E}_{f,D,x}[(\mathbb{E}_{D|f}[\hat{f}_2(x)]-\hat{f}_2(x))^2] $ is  $p_2-c(M_2)$, where (omitting superscripts from weights)
\begin{align*}
    p_2 &= -\mathbb{E}\left[\left(w_1\hat{f}_1(x)+w_2\hat{f}_2(x)-y\right)^2\right]+\mathbb{E}\left[\left(\hat{f}_1(x)-y\right)^2\right]
    \\ & = (1-w_1^2)( \alpha_1^2 B_0 +  V_1) - w_2^2(\alpha_2^2 B_0 + V_2) - 2w_1 w_2\alpha_1\alpha_2 B_0 \text{ by \Cref{p:bias_variance} and \cref{eq:common_bias}}.
\end{align*}
So the payoff under firm $2$'s optimal model choice is $$ \max_{w \in [0,1], M_2 \in \mathcal{M}_2} (1-w^2) ( \alpha_1^2 B_0 +  V_1) - (1-w)^2(\alpha_2^2 B_0 + V_2) - 2w(1-w)\alpha_1\alpha_2 B_0 -c(M_2)$$
where $w$ is the weight the consumer places on model $M_1$. Firm 2 will enter if this maximum is positive and will not enter if it is negative.

At $M_1 = M^*$, we can perturb $M_1$ such that the derivative of $\alpha_1^2B_0+V_1$ is zero but the derivative of $\alpha_1$ is positive. This decreases $$(1-w^2) ( \alpha_1^2 B_0 +  V_1) - (1-w)^2(\alpha_2^2 B_0 + V_2) - 2w(1-w)\alpha_1\alpha_2 B_0 - c(M_2)$$
for every feasible $w$ and $M_2$. Therefore this perturbation decreases the equilibrium payoff for firm $2$, which is the maximum over feasible $w$ and $M_2$ of this expression,  and so decreases the value of entry. So there exists an interval of $c_f$ such that firm 2 would enter if firm 1 chooses any $M_1$ with $\alpha_1 \leq \alpha_1^*$ (where $\alpha_1^*$ is the bias constant of $M^*$) but not if firm 1 chooses some $M_1$ with $\alpha_1>\alpha^*$. If the outside option is sufficiently low, firm 1 will choose such a model $M_1$ at equilibrium.
\end{proof}

\subsection{Proof of \Cref{p:sequential_cost}}

\begin{proof}
Suppose firm $1$ chooses model $M_1$ with expected variance $V_1$. Firm $2$'s payoff from entering is
\begin{align*}
    p_2 &  = \max_{w\in [0,1], M_2 \in \mathcal{M}_2} \left(\varphi(w,M_2) + (1-w^2)V_1 - c_m(M_2) - c_f\right)
    \\ & = (1-w^2)V_1-c_f + \max_{w\in [0,1], M_2 \in \mathcal{M}_2} \left(\varphi(w,M_2)  - c_m(M_2) \right)\end{align*}
for some function $\varphi(w,M_2)$ that is independent of firm $1$'s choice $M_1$. So fixing $c_f$, firm $2$ chooses to enter if $V_1$ is above a threshold level and not to enter if $V_1$ is below this threshold level.

Given $\epsilon>0$, we can choose $c_f$ such that this threshold level is $V_1^* +\epsilon$, where $V_1^*$ is the expected variance of model $M_1^*$. Taking $\epsilon$ sufficiently small, firm $2$ would enter if firm $1$ chooses model $M_1^*$ but not if firm $1$ chooses model $M_1'$. Then for $\underline{u}$ sufficiently negative, it is optimal for firm $1$ to choose a model such that firm $2$ does not enter, which requires $c(M_1)>c(M_1^*)$.
\end{proof}

\end{document}